\documentclass{LMCS}

\def\dOi{9(4:17)2013}
\lmcsheading%
{\dOi}
{1--28}
{}
{}
{Nov.~\phantom01, 2012}
{Dec.~\phantom05, 2013}
{}

\ACMCCS{[{\bf Theory of computation}]: Logic---Modal and temporal
  logics; [{\bf Computing methodologies}]: Artificial
  intelligence---Knowledge representation and reasoning---Reasoning
  about belief and knowledge}

\usepackage{tabu}
\usepackage{hyperref}

\usepackage{microtype}
\usepackage{txfonts}
\usepackage[all]{xy}
\usepackage{amssymb,amsmath}
\usepackage{xypic}
 \usepackage{amsmath}
  \usepackage{amsthm}
  \usepackage{latexsym}
  \usepackage{amscd}
  \usepackage{amssymb}
    \usepackage{tikz}

\usepackage{color}
\definecolor{darkgreen}{rgb}{0,0.5,0}
\definecolor{darkblue}{rgb}{0,0,0.8}
\definecolor{darkred}{rgb}{0.9,0,0}

\newtheorem{proposition}[thm]{Proposition}
\newtheorem{remark}[thm]{Remark}
\newtheorem{example}[thm]{Example}
\newtheorem{definition}[thm]{Definition}
\newtheorem{theorem}[thm]{Theorem}
\newtheorem{lemma}[thm]{Lemma}

\newcommand{\commment}[1]{}
\newcommand\val[1]{{\lbrack\!\lbrack} {#1}{\rbrack\!\rbrack}}
\newcommand{\bbA}{\mathbb{A}}
\newcommand{\bbB}{\mathbb{B}}

\newcommand{\jty}{J^{\infty}}
\newcommand{\mty}{M^{\infty}}

\newcommand{\f}{\mathcal{F}}
\newcommand{\asf}{\mathsf{a}}
\newcommand{\bsf}{\mathsf{b}}
\newcommand{\csf}{\mathsf{c}}
\newcommand{\hsf}{\mathsf{h}}
\newcommand{\isf}{\mathsf{i}}

\begin{document}
\title[Epistemic Updates on Algebras]{Epistemic Updates on Algebras}

\author[A.~Kurz]{Alexander Kurz\rsuper a}
\address{{\lsuper a}Department of Computer Science, University of Leicester\\
  Leicester, UK}
\email{ak155@mcs.le.ac.uk}

\author[A.~Palmigiano]{Alessandra Palmigiano\rsuper b}
\address{{\lsuper b}Institute for Logic, Language and Computation, University of Amsterdam\\
  Amsterdam, The Netherlands}
\email{a.palmigiano@uva.nl}



\begin{abstract}
  We develop the mathematical theory of epistemic updates with the
  tools of duality theory. We focus on the Logic of Epistemic Actions
  and Knowledge (EAK), introduced by Baltag-Moss-Solecki,
  without the common knowledge operator.
  We dually characterize the product update construction of EAK as a
  certain construction transforming the complex algebras associated
  with the given model into the complex algebra associated with the
  updated model.
  This dual characterization naturally generalizes to much wider
  classes of algebras, which include, but are not limited to,
  arbitrary BAOs and arbitrary modal expansions of Heyting algebras
  (HAOs).
  As an application of this dual characterization, we axiomatize the
  {\em intuitionistic} analogue of the logic of epistemic knowledge
  and actions, which we refer to as IEAK, prove soundness and
  completeness of IEAK w.r.t.\ both algebraic and relational models,
  and illustrate how IEAK encodes the reasoning of agents in a
  concrete epistemic scenario.\end{abstract}

\maketitle

\section{Introduction}
Duality theory is an established methodology in the mathematical theory of modal logic, and has been the driving engine of some of its core results  (e.g.\ the theory of canonicity), as well as of its generalizations (e.g.\ coalgebraic logics), and of extensions of techniques and results from modal logic to other nonclassical logics (e.g.\ Sahlqvist correspondence for substructural logics). Together with \cite{MPS}, the present paper is concerned with applying duality theory to a close cognate of modal logic, namely {\em Dynamic Epistemic Logic}, and starting to take stock of the results of this application. The dynamic epistemic logic considered in the present paper is the Logic of Epistemic Actions and Knowledge due to  Baltag-Moss-Solecki \cite{BMS}, and we refer to it as EAK.

The main feature of the relational semantics of EAK is the so-called
{\em product update} construction, which is grounded on a Kripke-style
encoding of epistemic actions.
Epistemic actions in this setting are formalized as {\em action
  structures}: finite pointed relational structures, each state of
which is endowed with a formula (its {\em precondition}). Epistemic
updates are transformations of the model encoding the current
epistemic setup of the given agents, by means of which the current
model is replaced with its {\em product update} with the action
structure.

In the present paper, the product update construction introduced in
\cite{BMS} is dually characterized as a certain construction
transforming the complex algebra associated with any given model into
the complex algebra associated with the model updated by means of a
given action structure. As is well known (see e.g.\ \cite[Chapter 5]{BdRV01}), these complex algebras are complete atomic BAOs (Boolean
algebras with operators). The dual characterization provided in the
present paper naturally generalizes to much wider classes of algebras,
which include, but are not limited to, arbitrary BAOs and modal
expansions of arbitrary Heyting algebras (HAOs). Thanks to this
construction, the benefits and the wider scope of applications given
by a point-free, nonclassical theory of epistemic updates are made
available: for instance, this construction provides the tools to
answer the question of how to define product updates on topological
spaces.

As an application of this dual characterization, we axiomatize the {\em intuitionistic} analogue of the logic of epistemic actions and knowledge, which we refer to as IEAK, prove soundness and completeness of IEAK w.r.t.\ both algebraic and relational models, and illustrate how IEAK encodes the reasoning of agents in a concrete epistemic scenario.

\bigskip Let us informally  expand on (a) how general principles in duality theory are applied to the Stone duality setting for the relational models of EAK, and yield an {\em algebraic characterization} of epistemic updates (this is the approach introduced in \cite{MPS} and applied there to epistemic actions of public-announcement type), and on (b) how the results of \cite{MPS} are extended from public announcements to general epistemic updates in the style of Baltag-Moss-Solecki.
In \cite{BMS}, given a relational model $M$ and an action structure $\alpha$, the  {\em product update} $M^\alpha$ is defined as a certain {\em submodel} of a certain {\em intermediate model} $M\times \alpha$, the domain of which is the cartesian product of the domains of $M$ and of $\alpha$. In the present paper, we preliminarily observe that the intermediate model $M\times \alpha$ can be actually identified with an appropriate (pseudo) {\em coproduct} $\coprod_{\alpha}M$ of $M$, indexed by the states of $\alpha$. Hence, the original product update construction can be understood as the concatenation of a certain coproduct-type construction, followed by a subobject-type construction, as illustrated by the following diagram:
\[
M\hookrightarrow \coprod_{\alpha}M \hookleftarrow M^\alpha.
 \]
 As is very well known (cf.\ e.g.\ \cite{DaPr}) in duality theory,
 coproducts can be dually characterized as products, and subobjects as
 quotients; an aspect of this dual characterization---which we use to
 our advantage and which is worth stressing at this point---is that,
 for these dual characterizations to be defined, an {\em a priori}
 specification of the fully fledged category-theoretic environment in
 which these constructions are taken is actually not needed; rather,
 the appropriate category-theoretic environment can be specified {\em
   a posteriori}, as long as these constructions can be recognized as
 products, subobjects, etc. For instance, the `subobject-type'
 construction on Kripke models mentioned above defines a proper
 subobject in the category of Kripke models and relation-preserving
 maps (the latter being dually characterized as {\em continuous
   morphisms}, see e.g.\ \cite{Ghi10}) and not in the standard
 category of Kripke models and p-morphisms. We do not expand on the
 category-theoretic account of these constructions further on. In the
 light of this understanding of dual characterizations, the
 construction of {\em product update} can be viewed as a ``subobject
 after coproduct'' concatenation, and is dually characterized on
 algebras by means of a ``quotient after product'' concatenation, as
 illustrated in the following diagram:
\[
\bbA\twoheadleftarrow \prod_{\alpha}\bbA \twoheadrightarrow \bbA^\alpha,
\]
resulting in the following two-step process. First, the coproduct
$\coprod_{\alpha}M$ is dually characterized as a certain {\em product}
$\prod_{\alpha}\bbA$, indexed as well by the states of $\alpha$, and
such that $\bbA$ is the algebraic dual of $M$; second, an appropriate
{\em quotient} of $\prod_{\alpha}\bbA$ is then taken, as an instance
of the general construction introduced in \cite{MPS} to account for
public announcements.
Note that again these constructions can be
interpreted in any category of algebras that supports the appropriate
notions of product and quotient.
This two-step process, taken as a whole,
modularly generalizes the dual characterization of \cite{MPS}: indeed,
public announcements can be encoded as certain one-state action
structures $\alpha$, in such a way that, for any given model $M$, its
corresponding intermediate model $M\times \alpha$ can be identified
with $M$. Hence, when instantiated to action structures encoding
public announcements, the two-step construction introduced in the
present paper can be identified with its second step, discussed in it
full generality in \cite{MPS}.  \commment{ and (modulo some minor
  tweaks accounting for the `pseudo's) falls naturally into the
  general duality pattern illustrated in the following table:
\begin{center}
\begin{tabular}{c c c}
Oppositions && Translations \\
\begin{tabular}{r|l}
Models & Algebras\\
 States/Objects & Elements/Properties \\
 Extensionality & Intensionality \\
\end{tabular}
& $\quad$ &
\begin{tabular}{r c l}
Embeddings & $\leftrightsquigarrow$& Surjections\\
Submodels & $\leftrightsquigarrow$& Quotients \\
Coproducts & $\leftrightsquigarrow$& Products \\
$M\hookrightarrow \coprod_{\alpha}M \hookleftarrow M^\alpha$ & $\leftrightsquigarrow$& $\bbA\twoheadleftarrow \prod_{\alpha}\bbA \twoheadrightarrow \bbA^\alpha$\\
\end{tabular}
\end{tabular}
\end{center}
}


\bigskip As mentioned early on, the advantage brought about by the dual characterization of product updates (which defines the {\em epistemic updates on algebras}) is that its definition naturally holds in much {\em more general} classes of algebras than the ones given by the algebras dually associated with the Kripke models. These more general classes include -- but are not limited to --  arbitrary BAOs, and modal expansions of arbitrary Heyting algebras (HAOs).

Exactly in the same way in which dynamic formulas in the language of EAK can be interpreted on relational models using the product update construction, the {\em algebraic} counterpart of this construction can be used to interpret the same formulas on {\em algebraic models}, i.e., tuples $(\bbA, V)$ consisting of algebras and assignments, such that the algebraic version of epistemic update is defined on $\bbA$.

For instance, based on Definition \ref{def: extension}, it is easy to see that the class of algebraic models based on {\em arbitrary} BAOs (which class  properly extends the class of complete and atomic BAOs) provides sound and complete {\em pointfree} semantics for EAK; moreover, as a straightforward consequence of this fact, epistemic updates can be defined on e.g.\ {\em descriptive general frames} via the classical Stone/J\'onsson-Tarski duality (we do not provide an explicit definition in the present paper).

But more generally, {\em each} class of algebraic models gives rise to {\em some} logic of epistemic actions and knowledge via the interpretation defined in 
 Definition \ref{def: extension}.
In particular, the set of axioms describing the behaviour of the intuitionistic dynamic connectives (cf.\ Section \ref{subsec:Axiomatization}) naturally arises from the class of algebraic models based on {\em Heyting algebras with operators} (HAOs) (which, for the sake of the present paper, are understood as Heyting algebras expanded with one normal $\Box$ operator and one normal $\Diamond$ operator). The axiomatization of HAOs does not imply the existence of any interaction between the static (epistemic) box and diamond operations, and of course, for the purpose of describing the epistemic setup of each agent, it is desirable to have at least as strong an axiomatization as one which forces the pairs of epistemic modal operators associated with each agent to be interpreted by means of {\em one and the same} relation. The intuitionistic basic modal logic IK \cite{FS84, Sim} is the weakest axiomatization which implies the desired connection between the modal operations; its canonically associated class of algebras is a subclass of HAO which we refer to as Fischer-Servi algebras, or {\em FS-algebras} (cf.\ Definition \ref{def: IEA}). The logic IEAK introduced in the present paper arises as the logic of epistemic actions and knowledge associated with the class of algebraic models  based on FS-algebras.

In fact, along with the mentioned definition, a second way to define IEAK is proposed in the present paper, which reflects the idea that the epistemic set-up of agents might be encoded by {\em equivalence} relations. To account for this possibility, Prior's MIPC \cite{Prior} can be alternatively adopted instead of IK as the underlying static logic of IEAK, and {\em monadic Heyting algebras} can be taken in place of the more general FS-algebras; the results presented in what follows develop these two options side by side in a modular way.

The structure of the paper goes as follows: Section 2 collects the needed preliminaries on classical EAK and intuitionistic modal logic. In Section 3, the dual, algebraic characterization of epistemic updates is introduced. In Section 4, the intuitionistic logic of epistemic actions and knowledge IEAK is axiomatically defined, as well as its interpretation on models based on Heyting algebras. Moreover, the relational semantics for intuitionistic modal logic/IEAK is described in detail. Finally, the soundness of IEAK is proved w.r.t.\ algebraic (hence relational) models, as well as the completeness of IEAK w.r.t.\ relational (hence algebraic) models.  
In Section 5, it is shown how IEAK can be used to describe and reason about a concrete epistemic scenario.
Details of  all the proofs in the mentioned sections are collected in Section 6, the appendix.

\section{Preliminaries}

\subsection{The logic of epistemic actions and knowledge}
\label{ssec:EAK}

In the present subsection, the relevant preliminaries on the syntax and semantics of the logic of epistemic actions and knowledge (EAK) \cite{BMS} will be given, which are different but equivalent to the original version appearing in \cite{BMS}; the aspects in which the account given here departs from the original version are intended to make the dualization construction  more transparent, which will be introduced in the following section.

Let \textsf{AtProp} be a countable set of proposition letters. The set $\mathcal{L}$ of formulas $\phi$ of (the  single-agent\footnote{The multi-agent generalization of this simpler version is straightforward, and consists in taking the indexed version of the modal operators, axioms and interpreting relations (both in the models and in the action structures) over a set of agents.} version of) the logic of epistemic actions and knowledge (EAK)   and the set $\mathsf{Act}(\mathcal{L})$ of the {\em action structures} $\alpha$ {\em over} $\mathcal{L}$ are built simultaneously  as follows:
\begin{center}
$\phi::= p\in \mathsf{AtProp} \mid \neg\phi \mid \phi\vee \phi\mid  \Diamond\phi\mid  \langle\alpha\rangle \phi\;\; (\alpha\in \mathsf{Act}(\mathcal{L})).$
\end{center} An {\em action structure over} $\mathcal{L}$ is a tuple  $\alpha = (K, k, \alpha, Pre_\alpha)$, such that $K$ is a finite nonempty set,  $k\in K$,  $\alpha\subseteq K\times K$ and $Pre_\alpha: K\to \mathcal{L}$. Notice that $\alpha$ denotes {\em both} the action structure {\em and} the accessibility relation of the action structure. Unless explicitly specified otherwise, occurrences of this symbol are to be interpreted contextually: for instance, in  $j\alpha k$, the symbol $\alpha$ denotes the relation; in  $M^{\alpha}$, the symbol $\alpha$ denotes the action structure. Of course, in the multi-agent setting, each action structure comes equipped with {\em a collection} of accessibility relations indexed in the set of agents, and then the abuse of notation disappears.

Sometimes we will write $Pre(\alpha)$ for $Pre_\alpha(k)$. Let $\alpha_i = (K, i, \alpha, Pre_\alpha)$ for every action structure $\alpha = (K, k, \alpha, Pre_\alpha)$ and every $i\in K$.
The standard stipulations hold for the defined connectives $\top$, $\bot$, $\wedge$, $\rightarrow$ and $\leftrightarrow$.

 Models for EAK are relational structures $M = (W, R, V)$ such that $W$ is a nonempty set, $R\subseteq W\times W$ and $V:\mathsf{AtProp}\to \mathcal{P}(W)$. The evaluation of the static fragment of the language is standard. For every Kripke frame $\mathcal{F} = (W, R)$ and every  $\alpha\subseteq K\times K$, let the Kripke frame $\coprod_{\alpha}\mathcal{F} : = (\coprod_{K}W, R\times \alpha)$ be defined\footnote{We will of course apply this definition to relations $\alpha$ which are part of the specification of some action structure; in these cases, the symbol $\alpha$ in $\coprod_{\alpha}\mathcal{F}$ will be  understood as the action structure. This is why the abuse of notation turns out to be useful.} as follows: $\coprod_{K}W$ is the $|K|$-fold coproduct of $W$ (which is set-isomorphic to $W\times K$), and $R\times \alpha$ is the binary relation on $\coprod_{K}W$ defined as $$(w, i)(R\times \alpha) (u, j)\quad \mbox{ iff }\quad w R u\ \mbox{ and }\ i\alpha j.$$

 For every model $M = (W, R, V)$ and every action structure $\alpha = (K, k, \alpha, Pre_\alpha)$, let $$\coprod_{\alpha}M := (\coprod_{K}W, R\times \alpha, \coprod_{K}V )$$ be such that its underlying frame is defined as detailed above, and $(\coprod_{K}V)(p): = \coprod_{K}V(p)$ for every $p\in \mathsf{AtProp}$. Finally, the {\em update} of $M$ with the action structure $\alpha$ is the submodel $M^\alpha: = (W^\alpha, R^\alpha, V^\alpha)$ of $\coprod_{\alpha}M$ the domain of which is the subset $$W^\alpha: = \{(w, j)\in \coprod_{K}W\mid M, w\Vdash Pre_\alpha(j)\}.$$

  Given the preliminary definition above, formulas of the form $\langle\alpha\rangle\phi$ are evaluated as follows:
$$M, w\Vdash \langle\alpha \rangle \phi\quad \mbox{ iff } \quad M, w\Vdash  Pre_\alpha(k)  \mbox{ and } M^\alpha, (w, k)\Vdash \phi.$$

\begin{proposition}[{\cite[Theorem 3.5]{BMS}}]
\label{prop: classical PAL}
EAK is axiomatized completely by the axioms and rules for the modal logic S5/IK plus the following axioms:
\begin{enumerate}
\item $\langle\alpha\rangle p\leftrightarrow (Pre(\alpha)\wedge p)$;
\item $\langle \alpha\rangle \neg\phi\leftrightarrow (Pre(\alpha)\wedge \neg\langle \alpha\rangle \phi )$;
\item $\langle \alpha\rangle (\phi\vee \psi) \leftrightarrow (\langle \alpha\rangle \phi\vee \langle \alpha\rangle \psi )$;
\item $\langle \alpha\rangle \Diamond\phi\leftrightarrow (Pre(\alpha)\wedge \bigvee\{\Diamond\langle\alpha_i\rangle\phi\mid k\alpha i\})$.
\end{enumerate}

where $\alpha_i = (K, i, \alpha, Pre_\alpha)$ for every action structure $\alpha = (K, k, \alpha, Pre_\alpha)$ and every $i\in K$.

\end{proposition}

\begin{remark}
 The intuitive understanding of action structures and of the product update construction has been extensively discussed in \cite{BMS}, by way of plenty of concrete examples; here we only limit ourselves to briefly report on some general pointers, and below we introduce a concrete scenario which will be then expanded on in Section 5.  An action structure encodes not only the {\em factual} information on a given action, but also its {\em epistemic} reflections on agents. Indeed, the designated action-state $k$ of $\alpha$ encodes the factual information; the other states in $K$ encode all its alternative appearances from the agents' viewpoint; in particular, $k\alpha i$ is to mean that the agent considers it possible that the action-state $i$ encodes the action which has been actually executed, instead of $k$. Correspondingly, $\alpha_i$ is the action structure which encodes this shift in the perception of the action actually executed, and public announcements are encoded as action structures with only the actual state $k$ which $\alpha$-accesses itself (since the agent entertains no doubts on what is actually happening). The product update construction builds on this intuition; copies of  $M$ are  created in as many colors as there are appearances of the action taking place; a copy of a given state of $M$ accesses a copy of one of its original successors (in the same or in another color) only if also the color of the copy of the successor is an  $\alpha$-successor of the color of the copy of the given state. Then all the copies of a given original state of $M$ are eliminated if the original state does not satisfy the preconditions of the execution of their respective color-appearance (which means that that particular transition could not have been executed in the first place under that particular state of affairs).
 \end{remark}

\begin{example}\label{exle:actions}
  The following example is based on a scenario that will be analysed
  in detail in Section~\ref{sec:illustration}. There is a set $I$ of
  three agents, $\mathsf{a}, \mathsf{b}, \mathsf{c}$, and three cards,
  two of which are white, and are each held by $\mathsf{b}$ and
  $\mathsf{c}$, and one is green, and is held by
  $\mathsf{a}$. Initially, each agent only knows the color of its own
  card, and it is common knowledge among the three agents that there
  are two white cards and one green one. Then $\mathsf{a}$ shows its
  card only to $\mathsf{b}$, but in the presence of $\mathsf{c}$. Then
  $\mathsf{b}$ announces that $\mathsf{a}$ knows what the actual
  distribution of cards is. Then, after having witnessed $\mathsf{a}$
  showing its card to $\mathsf{b}$, and after the ensuing public
  announcement of $\mathsf{b}$, agent $\mathsf{c}$ knows what the
  actual distribution is.



  For the sake of this scenario, we can restrict the set of
  proposition letters to $\{W_{\isf}, G_{\isf}\mid
  \isf\in I\}$. The intended meaning of $W_{\isf}$ and
  $G_{\isf}$ is `agent $\isf$ holds a white card', and
  `agent $\isf$ holds a green card' respectively.

  The action structure $\alpha$ encoding the action performed by agent
  $\mathsf{a}$ can be assimilated to the atomic proposition
  $G_{\mathsf{a}}$ being announced to the subgroup $\{\mathsf{a},
  \mathsf{b}\}$.
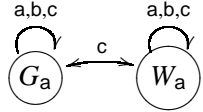
\begin{equation*}
\entrymodifiers={++[o][F-]}
\xymatrix{
*++[o][]{} 
& G_{\mathsf{a}} \ar@(ul,ur)^{\mathsf{a},\mathsf{b},\mathsf{c}}  \ar@{<->}@/^/[r]^{\mathsf{c}}
& W_{\mathsf{a}}\ar@(ul,ur)^{\mathsf{a},\mathsf{b},\mathsf{c}}
}
\end{equation*}
Formally, $\alpha = (K, k, \alpha_{\mathsf{a}}, \alpha_{\mathsf{b}},
\alpha_{\mathsf{c}}, Pre_\alpha)$ is specified as follows: $K = \{k,
l\}$; moreover, $Pre(\alpha) = Pre_\alpha(k) = G_{\mathsf{a}}$, and
$Pre(\alpha_l) = Pre_\alpha(l) = W_{\mathsf{a}}$; finally,
$\alpha_{\mathsf{a}} = \alpha_{\mathsf{b}} = \Delta_K$ and
$\alpha_{\mathsf{c}} = K\times K$.

To illustrate the update mechanism assume that the model $M$
is specified by
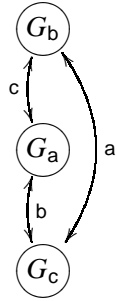
\begin{equation*}
\entrymodifiers={++[o][F-]}
\xymatrix{
*++[o][]{} & G_{\mathsf{b}}\ar@{<->}@/^20pt/[dd]^{\mathsf{a}}\\
*++[o][]{} & G_{\mathsf{a}}\ar@{<->}@/^/[u]^{\mathsf{c}}\ar@{<->}@/_/[d]^{\mathsf{b}}  \\
*++[o][]{} & G_{\mathsf{c}}
}
\end{equation*}
where we omitted the self-loops corresponding to epistemic uncertainty
being reflexive.
Then $\coprod_\alpha M$ is depicted by
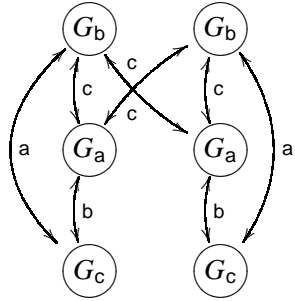
\begin{equation*}
\entrymodifiers={++[o][F-]}
\xymatrix{
*++[o][]{}
& G_{\mathsf{b}}\ar@{<->}@/_30pt/[dd]^{\mathsf{a}}
                           \ar@{<->}@/_/[dr]_{\mathsf{c}}
& G_{\mathsf{b}}\ar@{<->}@/^20pt/[dd]^{\mathsf{a}}\\
*++[o][]{}
& G_{\mathsf{a}}
\ar@{<->}@/^/[u]_{\mathsf{c}}\ar@{<->}@/^/[ur]^{\mathsf{c}}
\ar@{<->}@/_/[d]^{\mathsf{b}}
& G_{\mathsf{a}}
\ar@{<->}@/^/[u]_{\mathsf{c}}\ar@{<->}@/_/[d]^{\mathsf{b}}  \\
*++[o][]{} & G_{\mathsf{c}}& G_{\mathsf{c}}
}
\end{equation*}
and the product-update $M^\alpha$ is
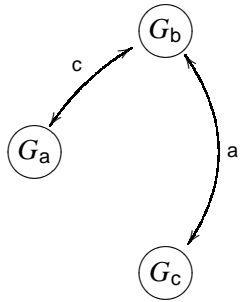
\begin{equation*}
\entrymodifiers={++[o][F-]}
\xymatrix{
*++[o][]{}
& *++[o][]{}
& G_{\mathsf{b}}\ar@{<->}@/^20pt/[dd]^{\mathsf{a}}\\
*++[o][]{}
& G_{\mathsf{a}}
\ar@{<->}@/^/[ur]^{\mathsf{c}}
& *++[o][]{}  \\
*++[o][]{} & *++[o][]{} & G_{\mathsf{c}}
}
\end{equation*}
where the two states marked $G_{\mathsf{b}}, G_{\mathsf{c}}$ in the
left-hand column get deleted because our scenario induces the
assumptions $G_{\mathsf{b}}\wedge G_{\mathsf{a}} = \bot =
G_{\mathsf{c}}\wedge G_{\mathsf{a}}$. Similarly, the state marked
$G_{\mathsf{a}}$ in the right-hand column disappears because of
$G_{\mathsf{a}}\wedge W_{\mathsf{a}}=\bot$.

The action structure $\beta$ encoding the public announcement
performed by agent $\mathsf{b}$ can be specified as a one-state
structure such that $\alpha_{\isf} = \Delta_{K}$ for each
$\isf\in I$, and the precondition of which is the formula $Pre
(\beta) = \bigwedge_{\isf \in I}(G_{\isf}\rightarrow
\Box_{\mathsf{a}} G_{\isf})$. Accordingly, updating
$M^\alpha$ with $\beta$ yields the model
\begin{equation*}
\entrymodifiers={++[o][F-]}
\xymatrix{
*++[o][]{}
& G_{\mathsf{a}}
}
\end{equation*}
according to which all agents know the distribution of the cards (since
there is only one state and, thus, no epistemic uncertainty).  In
Section~\ref{sec:illustration}, we will show that the reasoning
in this scenario can be syntactically formalized on an
intuitionistic base by (the appropriate multi-agent version of) the
logic IEAK introduced in Section 4.1.
\end{example}

\subsection{The intuitionistic modal logics MIPC and IK}
\label{ssec:IK}

Respectively introduced by Prior with the name MIPQ \cite{Prior}, and
by Fischer-Servi \cite{FS84}, the two intuitionistic modal logics the
present subsection focuses on are largely considered the
intuitionistic analogues of S5 and of K, respectively. These logics
have been studied by many authors, viz.\ \cite{Bez98, Bez99, Sim} and
the references therein. In the present subsection, the notions and
facts needed for the purposes of the present paper will be briefly
reviewed. The reader is referred to \cite{Bez98, Bez99, Sim} for their
attribution.  The formulas for both logics are built by the following
inductive rule (and let $\mathcal{L}_{IK}$ denote the resulting set of
formulas):

\begin{center}
$\phi::= \bot \mid p\in \mathsf{AtProp} \mid \phi\wedge\psi \mid \phi\vee \psi\mid  \phi\to \psi\mid \Diamond\phi \mid  \Box \phi.$
\end{center}
Let $\top$ be defined as $\bot\to \bot$ and, for all formulas $\phi$ and $\psi$, let $\neg \phi$ be defined as $\phi\to \bot$ and $\phi\leftrightarrow \psi$ be defined as $(\phi\to \psi)\wedge(\psi\to \phi)$.
The logic IK is
the smallest set of formulas in the language above which contains all the axioms of intuitionistic propositional
logic, the following modal axioms
\begin{enumerate}[label=FS\arabic*.]
\item[]$\Box (p \to q)\to (\Box p \to\Box q),$
\item[]$\Diamond(p\vee q)\to(\Diamond p\vee \Diamond q),\ \neg\Diamond\bot$,
\item[FS1.]$\Diamond (p \to q)\to(\Box p \to \Diamond q),$
\item[FS2.]$(\Diamond p\to \Box q)\to \Box(p\to q),$
\end{enumerate}
and is closed under substitution, modus ponens and necessitation $(\vdash \varphi/\vdash\Box \varphi)$.
 The logic MIPC is
the smallest set of formulas in the language above which contains all the axioms of intuitionistic propositional
logic, the following modal axioms
\begin{itemize}
\item[]$\Box p \to p,\ p \to \Diamond p,$
\item[]$\Box (p \to q)\to (\Box p \to\Box q),\ \Diamond(p\vee q)\to(\Diamond p\vee \Diamond q),$
\item[]$\Diamond p \to\Box \Diamond p,\  \Diamond\Box p\to \Box p,$
\item[]$\Box (p \to q)\to(\Diamond p \to \Diamond q),$
\end{itemize}
and is closed under substitution, modus ponens and necessitation $(\vdash \varphi\,/\,{\vdash}\Box \varphi)$.

The relational structures for IK (resp.\ MIPC), called {\em IK-frames} (resp.\ {\em MIPC-frames}), are triples $\f = (W, \leq, R)$ such that $(W, \leq)$ is a nonempty poset and $R$ is a binary (equivalence) relation such that
$$(R\circ{\geq})\subseteq ({\geq}\circ R),\quad  \quad({\leq}\circ R)\subseteq (R\circ {\leq}), \quad \quad R = ({\geq}\circ R)\cap (R\circ {\leq}).$$
where $\circ$ denotes composition written in the usual relational order.
Notice that, in the case of MIPC-frames,  $R$ being symmetric implies that the second condition is equivalent to the first one, and the third condition is equivalent to $R = (R\circ {\leq})$.
 \commment{
The {\em p-morphisms} $f: \f\to \f'$ of IK-frames are maps $f: W\to W'$ satisfying the following conditions: for all $x, y\in W$ and every $z\in W'$,\\
M1. if $x \leq y$, then $f(x) \leq f(y)$;\\
M2. if $f(x) \leq z$, then $f(x') = z$ for some $x' \leq x$;\\
M3. if $xRy$, then $f(x)R'f(y)$;\\
M4. if $f(x)R'z$, then $x R y$ for some $y$ such that $f(y) \leq z$;\\
M5. if $f(x)({\geq} \circ R')z$, then $x  ({\geq} \circ R)y$ for some $y$ such that $z\leq f(y)$.\\
}
{\em IK-models} (resp.\ {\em MIPC-models}) are structures $M = (\f, V)$ such that $\f$ is an IK-frame (resp.\ an MIPC-frame) and $V: \mathsf{AtProp}\to \mathcal{P}^\downarrow(W)$ is a function mapping proposition letters to downward-closed subsets of $W$, where, for every poset $(W, \leq)$, a subset $Y$ of $W$ is {\em downward-closed} if for every $x, y\in W$, if $x\leq y$ and $y\in Y$ then $x\in Y$. For any such model, its associated extension map $\val{\cdot}_M: \mathcal{L}_{IK}\to \mathcal{P}^\downarrow(W)$ is defined recursively as follows:
\begin{center}
\begin{tabular}{r c l}
$\val{p}_M$ &$=$& $V(p)$\\
$\val{\bot}_M$&$ =$&$ \varnothing$\\
$\val{\phi\vee \psi}_M$ &$=$& $\val{\phi}_M\cup\val{\psi}_M$\\
$\val{\phi\wedge \psi}_M$ &$=$& $\val{\phi}_M\cap\val{\psi}_M$\\
$\val{\phi\to \psi}_M$ &$=$& $(\val{\phi}_M\cap \val{\psi}_M^c){\uparrow}^c$\\
$\val{\Diamond\phi}_M$ &$=$& $R^{-1}[\val{\phi}_M]$\\
$\val{\Box\phi}_M$ &$=$& $(({\geq}\circ R)^{-1}[\val{\phi}_M^c])^c$\\
\end{tabular}
\end{center}
where $(.)^c$ is the complement operation. For any model $M$ and any formula $\phi$, we write:
\begin{itemize}
\item[] $M, w\Vdash \phi$ if $w\in \val{\phi}_M$;
\item[] $M\Vdash \phi$ if $\val{\phi}_M = W$;
\item[] $\f\Vdash \phi$ if $\val{\phi}_M = W$ for any model $M$ based on $\f$.
\end{itemize}

\begin{proposition}\label{MIPC s&c}  IK (resp.\ MIPC) is sound and complete with respect to the class of IK-frames (resp.\ MIPC-frames).\end{proposition}

The algebraic semantics for IK (MIPC) is given by a variety of Heyting algebras with operators (HAOs) which are called Fischer-Servi algebras (monadic Heyting algebras):
\begin{definition}
\label{def: IEA}
The algebra $\bbA = (A, \wedge, \vee, \to, \bot, \Diamond, \Box)$ is a {\em Fischer-Servi algebra} (FSA) if  $(A, \wedge, \vee, \to, \bot)$ is a Heyting algebra and the following inequalities hold:
\begin{itemize}
\item[] $\Box (x \to y)\leq \Box x \to\Box y$,
\item[] $\Diamond(x\vee y)\leq(\Diamond x\vee \Diamond y)$, $\Diamond\bot\leq \bot$,
\item[]$\Diamond (x \to y)\leq\Box x \to \Diamond y,$
\item[] $\Diamond x \to \Box y\leq \Box (x \to y)$.
\end{itemize}
The algebra $\bbA$ is a {\em monadic Heyting algebra} (MHA) if  $(A, \wedge, \vee, \to, \bot)$ is a Heyting algebra and the following inequalities hold:
\begin{itemize}
\item[] $\Box x \leq x$, $x \leq \Diamond x$;
\item[] $\Box (x \to y)\leq \Box x \to\Box y$,  $\Diamond(x\vee y)\leq(\Diamond x\vee \Diamond y)$;
\item[] $\Diamond x \leq\Box \Diamond x$, $\Diamond\Box x\leq \Box x$;
\item[] $\Box (x \to y)\leq\Diamond x \to \Diamond y$.
\end{itemize}
\end{definition}

\noindent It is well known and can be readily verified that every monadic Heyting algebra is an FS-algebra.
The inequalities above can be equivalently written as equalities, thanks to the fact that, in any Heyting algebra, $x\leq y$ iff $x\to y = \top$. Clearly, any formula in the language $\mathcal{L}$ of IK (MIPC) can be regarded as a term in the algebraic language of FSAs (MHAs). Therefore, given an algebra $\bbA$ and an interpretation $V: \mathsf{AtProp}\to \bbA$,   an $\mathcal{L}$-formula $\phi$ is {\em true} in  $\bbA$ under the interpretation $V$ (notation: $(\bbA, V)\models \phi$) if the unique homomorphic extension of $V$, denoted by $\val{\cdot}_V: \mathcal{L}\to \bbA$,  maps $\phi$ to $\top^{\bbA}$. An $\mathcal{L}$-formula is {\em valid} in $\bbA$ (notation: $\bbA\models \phi$), if $(\bbA, V)\models \phi$ for every interpretation $V$.

IK-frames give rise to complex algebras, just as Kripke frames do: for any IK-frame $\f$, the {\em complex algebra} of $\f$ is
$$\f^{+} = (\mathcal{P}^\downarrow(W), \cap,  \cup, \Rightarrow, \varnothing, \langle R\rangle, [{\geq}\circ R]),$$
where for all $X, Y\in \mathcal{P}^\downarrow(W)$, $$\langle R\rangle X = R^{-1}[X], \ \ [{\geq}\circ R]X = (({\geq}\circ R)^{-1}[X^c])^c, \ \  X\Rightarrow Y = (X\cap Y^c){\uparrow}^c.$$
Clearly, given a model $M = (\f, V)$, the extension map $\val{\cdot}_M: \mathcal{L}\to \f^+$ is the unique homomorphic extension of $V: \mathsf{AtProp}\to \f^+$.

\begin{proposition}
\label{prop:algebras and frames}
For every IK-model $(\f, V)$ and every  $\mathcal{L}$-formula $\phi$,
\begin{enumerate}
\item $(\f, V)\Vdash \phi$ iff $(\f^+, V)\models \phi$.
\item $\f^{+}$ is an FS-algebra.
\item If $R$ is an equivalence relation, then $\f^{+}$ is a monadic Heyting algebra.
\end{enumerate}
\end{proposition}

\section{Epistemic updates on algebras}
\label{sec:EUA}
In Section \ref{ssec:EAK}, for every model $M$ and every action $\alpha$ over $\mathcal{L}$, the updated model $M^\alpha$ was defined as a submodel of the intermediate structure $\coprod_{\alpha}M$. In the present section, this construction  is dually characterized on algebras in two steps: first dualizing the construction procedure of $\coprod_{\alpha}M$, and then taking an appropriate quotient of it.

We preliminarily disregard the logic, and define, for every algebra $\bbA$,  an {\em action structure over} $\bbA$ as a tuple $a = (K, k, \alpha, Pre_a)$ such that $K$ is a finite nonempty set, $k\in K$,  $\alpha\subseteq K\times K$ and $Pre_a: K\to \bbA$. The letters $b, c$ will typically denote elements of the algebras $\bbA$, and we will reserve the letter $a$ for action structures over algebras. Clearly, for every EAK-model $M$, each action structure $\alpha = (K, k, \alpha, Pre_\alpha)$ over $\mathcal{L}$ induces a corresponding action structure $a$ over the complex algebra $\bbA$ of the underlying frame of $M$, via the valuation $V:\mathcal{L}\to \bbA$ of $M$ (here identified with its unique homomorphic extension): namely, $a$ is defined as $a = (K, k, \alpha, Pre_a)$, with $Pre_a = V\circ Pre_\alpha$. Moreover, for every Kripke frame $\f = (W, R)$, and every action structure $a = (K, k, \alpha, Pre_a)$ over the complex algebra of $\f$, the intermediate structure can be defined as $\coprod_a\f: = (\coprod_{K}W, R\times \alpha)$, and the updated frame structure $\f^a$ can be defined as the subframe of $\coprod_a\f$ the domain of which is the subset $$W^a: = \{(w, j)\in \coprod_{K}W\mid w\in Pre_a(j)\}.$$

\subsection{Dually characterizing the intermediate structure}

For every algebra $\bbA$ and every action structure $a = (K, k,
\alpha, Pre_a)$ over $\bbA$, let $\prod_{a}\bbA$ be the $|K|$-fold
product of $\bbA$, which is set-isomorphic to the collection
$\bbA^{K}$ of the set maps $f:K\to \bbA$. The set $\bbA^K$ can be
canonically endowed with the same algebraic structure as $\bbA$ by
pointwise lifting the operations on $\bbA$; as such, it satisfies the
same equations as $\bbA$; however, in the cases in which $\bbA$ is the
complex algebra of some frame $\f = (W, R)$, the lifted modal
operators on $\bbA^K$ would not adequately serve as the algebraic
counterparts of the accessibility relation $(R\times \alpha)$ of the frame
$\coprod_{a}\f$, because they would only depend on $\bbA$, and not on
$a$. Therefore, alternative definitions are called for, which are provided at the end of the following discussion.

The picture below shows $\coprod_a \f$ if $a$
has two states.
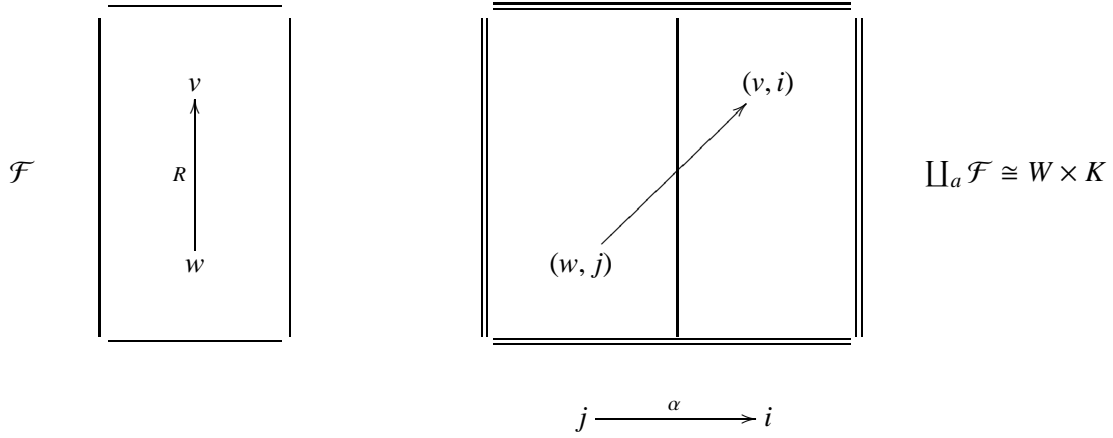
\begin{equation*}
\xymatrix@R=18pt@C=18pt{
& \ar@{-}[dddd]\ar@{-}[rr]&   & \ar@{-}[dddd]&
 & \ar@{=}[dddd] \ar@{=}[rrrr]& & \ar@{-}[dddd]& & \ar@{=}[dddd]&\\
&& v &&
 & & & &  {(v,i)} &&\\
\f&& & & \hspace{5ex}
 & & & & & &\coprod_a\f\cong W\times K \\
&& {\hspace{1.6ex}w\hspace{1.6ex}}\ar[uu]^{R}& &
 && {(w,j)}\ar[rruu]& & && \\
&\ar@{-}[rr]&&&
 & \ar@{=}[rrrr]& & & &&\\
&&&&
 & &j \ar[rr]^{\alpha}& & i&&\\
&&& &
 &&& & &&
}
\end{equation*}

\noindent As mentioned early on, the accessibility relation on
$\coprod_a\f\cong W\times K$ is the relation $(R\times \alpha)$
defined as follows:
$$(w, j) (R\times \alpha) (v, i) \textrm{ \ iff \ } j\alpha i \textrm{ and } wRv.$$
Hence, as usual, the operation $\Diamond$ on the complex algebra $\mathcal{P}(\coprod_\alpha \f)\cong \mathcal{P}(W\times K)$  is to be defined by taking  $(R\times \alpha)$-inverse images; that is, for any $f\subseteq W\times K$,
\begin{equation}
\label{eq:pointset diam-prod}
(w,j)\in  \Diamond f \textrm{ \ iff \ }
 wRv \textrm{ and } j\alpha i \textrm{ for some } (v,i)\in f.
 \end{equation}
 Via the following chain of isomorphisms,
\begin{equation}\label{eq:2WK}
\mathcal{P}(W\times
 K) =  2^{W\times K}\cong {2^W}^K =  \mathcal{P}(W)^K
\end{equation}
the subset $f$ can be equivalently represented as a map $f:K\to
\mathcal{P}(W)$, and consequently, the operation $\Diamond$ on
$\mathcal{P}(W\times K)$ can be equivalently represented as an
operation $\Diamond$ on $\mathcal{P}(W)^K$. Hence, condition
(\ref{eq:pointset diam-prod}) can be equivalently reformulated as
follows:
$$w \in (\Diamond f)(j) \textrm{ \ iff \ } w\in \Diamond^{\mathcal{P}(W)}(f(i)) \textrm{ for some } i\textrm{ such that } j\alpha i,$$
which is equivalent to the following identity holding in  $\mathcal{P}(W)^K$:
\begin{equation}
\label{eq:pointfree diam-prod}
(\Diamond f)(j) = \bigcup\{ \Diamond^{\mathcal{P}(W)}(f(i)) \mid j\alpha i\}.
\end{equation}
The argument above consists of a series of equivalent rewritings of
one initial condition involving the membership relation, and pivots on
the natural isomorphism \eqref{eq:2WK}.
These rewritings are aimed at expressing the initial condition
(\ref{eq:pointset diam-prod}) in a point-free way not
involving membership. The advantage of (\ref{eq:pointfree diam-prod})
over (\ref{eq:pointset diam-prod}) is that (\ref{eq:pointfree
  diam-prod}) applies much more generally than to powerset
algebras: 
namely, it applies to any join-semilattice $\bbA$ expanded with a unary operation
$\Diamond^\bbA$. For any such $\bbA$, and any action structure $a = (K, k, \alpha, Pre_a)$ over $\bbA$,
corresponding operations $\Diamond^{\prod_a\bbA}$ and
$\Box^{\prod_a\bbA}$ can be defined on the product $\prod_a\bbA$ as follows:
for every $f: K\to \bbA$,
let $\Diamond^{\prod_a\bbA}f: K\to \bbA$ and
$\Box^{\prod_a\bbA}f: K\to \bbA$ be given, for every $j\in K$, by
\begin{align}\label{eq:diam-prod}
(\Diamond^{\prod_a\bbA}f)(j) &= \bigvee\{\Diamond^\bbA f(i)\mid j\alpha i\}\\
\label{eq:box-prod}
(\Box^{\prod_a\bbA}f)(j) &= \bigwedge\{\Box^\bbA f(i)\mid j\alpha i\}.\end{align}
The series of equivalent rewritings  given above is an example of dual characterization; another such example appears in \cite[Section 3]{MPS}, and one more will be given in Section \ref{ssec:models}, which will serve to define the interpretation of dynamic epistemic formulas on algebraic models. The dual characterization above proves the following proposition:
\begin{proposition}
\label{prop: diamond pseudo product okay}
Let $\bbA$ be the complex algebra of some classical frame $\f = (W, R)$, and let $a = (K, k, \alpha, Pre_a)$ be an action structure over $\bbA$. Then the modal algebra $(\prod_a\bbA, \Diamond^{\prod_a\bbA})$ is isomorphic to the complex algebra of the intermediate structure $\coprod_{a}\f$.
\end{proposition}

The next proposition immediately follows from clauses (\ref{eq:diam-prod}) and (\ref{eq:box-prod}):
\begin{proposition}
\label{prop:product modalities normal}
For every lattice expansion $(\bbB, \Diamond, \Box)$, and every  action structure $a$ over $\bbA$,
 \begin{enumerate}
 \item if $\Diamond$ and $\Box$ are normal modal operators,  then $\Diamond^{\prod_a\bbA}$ and $\Box^{\prod_a\bbA}$  are normal modal operators.
 \item If $\bbB$ is a BA and $\Box: =\neg\Diamond\neg$, then $\Box^{\prod_a\bbA} = \neg \Diamond^{\prod_a\bbA}\neg$.
 \end{enumerate}
\end{proposition}

The discussion above justifies the following notation: in the remainder of the present paper, for every lattice expansion $\bbA = (\bbB, \Diamond, \Box)$ and every action structure $a$ over $\bbA$, the symbol $\prod_a\bbA$ will denote the algebra $(\prod_a\bbB, \Diamond^{\prod_{a}\bbA}, \Box^{\prod_{a}\bbA})$.

\begin{remark}
As discussed in Section 2.1, public announcements can be represented as those action structures $(K, k, \alpha, Pre_\alpha)$ over $\mathcal{L}$ such that $K$ is a one-element set, and $\alpha = \Delta_{K}$. Thus, each such action structure can be identified with the (publicly announced) formula $Pre_{\alpha}(\ast)$. Public announcement-type action structures $a$ over algebras $\bbA$ can be defined in an analogous way, and again identified with elements of $\bbA$. Then it is straightforward to see that the algebra
$\prod_{a}\bbA$ can be identified with the original algebra $\bbA$ when $a$ is a public announcement-type action structure. The same observation also holds in the more meaningful multi-agent setting.
\end{remark}

\subsection{Intermediate structures of FSAs, MHAs and of tense HAOs}

An HA $\bbB$ expanded with normal modal operations $(\bbB, \Diamond, \Box, \Diamondblack, \blacksquare)$ is a {\em tense} HAO if both $\Diamond$ and $\blacksquare$, and $\Diamondblack$ and $\Box$ are {\em adjoint pairs}, i.e.\ for all $b, c\in \bbA$, $$\Diamond b\leq c\quad \mbox{ iff }\quad b\leq \blacksquare c\quad\quad\mbox{ and }\quad\quad \Diamondblack b\leq c\quad \mbox{ iff }\quad b\leq \Box c.$$
We denote these adjunction relations by writing $\Diamond\dashv \blacksquare$ and $\Diamondblack\dashv \Box$. For any such tense HAO,   the algebra $(\prod_{a}\bbB, \Diamond^{\prod_{a}\bbA}, \Box^{\prod_{a}\bbA}, \Diamondblack^{\prod_{a}\bbA}, \blacksquare^{\prod_{a}\bbA})$ is defined as follows: $\Diamond^{\prod_{a}\bbA}$ and $\Box^{\prod_{a}\bbA}$ are defined as in the previous subsection, whereas, for every  $f: K\to \bbA$, let $\Diamondblack^{\prod_\alpha\bbA}f: K\to \bbA$ and $\blacksquare^{\prod_\alpha\bbA}f: K\to \bbA$ are respectively defined as follows: for every $j\in K$, $$(\Diamondblack^{\prod_\alpha\bbA}f)(j) = \bigvee\{\Diamondblack^\bbA f(i)\mid i\alpha j\},$$
$$(\blacksquare^{\prod_\alpha\bbA}f)(j) = \bigwedge\{\blacksquare^\bbA f(i)\mid i\alpha j\}.$$

\begin{proposition}
\label{fct:pseudo products of MHAs}
For every algebra $\bbA = (\bbB, \Diamond, \Box)$ and every action structure $a = (K, k, \alpha, Pre_a)$ over $\bbA$,
\begin{enumerate}
\item if $\bbA$ is an MHA and $\alpha$ is an equivalence relation, then $\prod_{a}\bbA$ is an MHA.
\item If $\bbA$ is an FSA, then $\prod_{a}\bbA$ is an FSA.
\item If $(\bbB, \Diamond, \Box, \Diamondblack, \blacksquare)$ is a tense HAO, then $(\prod_{a}\bbB, \Diamond^{\prod_{a}\bbA}, \Box^{\prod_{a}\bbA}, \Diamondblack^{\prod_{a}\bbA}, \blacksquare^{\prod_{a}\bbA})$ is a tense HAO.
\end{enumerate}
\end{proposition}

\begin{proof}
1. Since by assumption $\bbB$ is a HA, $\prod_{a}\bbB$ is a HA, so we only need to show the validity of the modal axioms. Throughout the proof, fix
$b, c\in \prod_{a}\bbB$. For the sake of readability, $\Diamond$ and $\Box$ will both denote the operations in $\bbA$ and in $\prod_a\bbA$ and are to be understood contextually: for instance, for every $j\in K$, the symbol $(\Diamond b)(j)$ is to be understood as $\pi_j(\Diamond^{\prod_{a}\bbA}(b))$, where
\begin{equation}\label{eq:def-pij}
\pi_j:\prod_{a}\bbA\to \bbA
\end{equation}
 is the projection on the $j$-indexed coordinate; the symbol $\Diamond b(j)$ is to be understood as $\Diamond^{\bbA}(\pi_j(b))$.\\
To prove that $b\leq \Diamond b$, we need to show that $b(j)\leq (\Diamond b)(j)$ for every $j\in K$, i.e.\ that
$b(j)\leq \bigvee\{\Diamond b(i)\mid j\alpha i\}$. Because $\alpha$ is reflexive and $\bbA$ is a MHA, we have:
\begin{center}
$b(j)\leq \bigvee\{b(i)\mid j\alpha i\}\leq \bigvee\{\Diamond b(i)\mid j\alpha i\}.$
\end{center}
The proof that $\Box b\leq b$ is order dual to the argument above.\\
To prove that $\Diamond b\leq \Box\Diamond b$, we need to show that $(\Diamond b)(j)\leq (\Box\Diamond b)(j)$ for every $j\in K$, i.e.\ that
\begin{center}
$\bigvee\{\Diamond b(i)\mid j\alpha i\}\leq \bigwedge\{\Box(\bigvee\{\Diamond b(h)\mid i\alpha h\})\mid j\alpha i\}$.
\end{center} It is enough to show that for each $j, i\in K$ such that $j\alpha i$,
$\Diamond b(i)\leq \Box(\bigvee\{\Diamond b(h)\mid i\alpha h\})$.
Because $\alpha$ is reflexive, we have:
\begin{center}
$\Diamond b(j)\leq \Box\Diamond b(j)\leq \Box(\bigvee\{\Diamond b(h)\mid i\alpha h\}).$
\end{center}
To prove that $\Diamond\Box b\leq \Box b$, we need to show that $(\Diamond\Box b)(j)\leq (\Box b)(j)$ for every $j\in K$, i.e.\ that
\begin{center}
$\bigvee\{\Diamond(\bigwedge\{\Box b(h)\mid i\alpha h\})\mid j\alpha i\}\leq \bigwedge\{\Box b(i)\mid j\alpha i\}$.
\end{center} It is enough to show that for each $j, i, i'\in K$ such that $j\alpha i$ and $j\alpha i'$,
$\Diamond (\bigwedge \{\Box b(h)\mid i'\alpha h\})\leq \Box b(i)$.
Because $\alpha$ is symmetric and transitive, we have $i'\alpha i$, hence:
\begin{center}
$\Diamond (\bigwedge \{\Box b(h)\mid i'\alpha h\})\leq \Diamond\Box b(i)\leq \Box b(i).$
\end{center}
The remaining verifications are left to the reader.\\
2. Similar to 1.\\
3.  For all $b, c\in \prod_{a}\bbB$,
\begin{center}
\begin{tabular}{r c l}
$\Diamond^{\prod_a\bbB}b\leq c$ & iff & $\bigvee\{\Diamond b(i)\mid j\alpha i\}\leq c(j)$ for every $j\in K$\\
& iff & $\Diamond b(i)\leq c(j)$ for every $j\in K$ and every $i\in K$ such that $j\alpha i$\\
& iff & $ b(i)\leq \blacksquare c(j)$ for every $i\in K$ and every $j\in K$ such that $j\alpha i$\\
& iff & $ b(i)\leq \bigwedge\{ \blacksquare c(j)\mid j\alpha i\}$ for every $i\in K$\\
& iff & $ b(i)\leq (\blacksquare^{\prod_a\bbB} c)(i)$ for every $i\in K$\\
& iff & $ b\leq \blacksquare^{\prod_a\bbB} c$.\\
\end{tabular}
 \end{center}
 The remaining adjunction relation is shown analogously.
\end{proof}
\subsection{Quotient of the intermediate structure}
\label{subsec: updates as pseudo quotients}
Throughout the present subsection, and unless specified otherwise, let $\bbA$ be a $\wedge$-semilattice  and let $a = (K, k, \alpha, Pre_a)$ be an action structure over $\bbA$.  Define the following equivalence relation $\equiv_a$ on $\prod_{a}\bbA$: for every $f, g\in \bbA^{K}$,

$$f\equiv_a g\ \mbox{ iff }\ f\wedge Pre_a = g\wedge Pre_a.$$
Let $[f]_a$ be the equivalence class of $f\in \bbA^{K}$. Usually, the subscript will be dropped when there is no risk of confusion.
Let the quotient set $\bbA^{K}/{\equiv_a}$ be denoted by $\bbA^a$.

The properties of this quotient are well known, and a detailed account of them can be found in \cite[Section 3.1]{MPS}, in a setting in which $\prod_{a}\bbA$ and $Pre_a$ respectively generalize to an arbitrary algebra  and  to an arbitrary element  of that algebra. In the remainder of this subsection, we will report on the relevant facts and properties, specialized to the present context, referring the reader to \cite{MPS} for proofs.

Clearly, $\bbA^a$ is an ordered set by putting $[b]\leq [c]$ iff $b'\leq_\bbA c'$ for some $b'\in [b]$ and some $c'\in [c]$. Let
\begin{equation}\label{eq:def-pi}
\pi = \pi^a:\prod_a\bbA\to \bbA^a
\end{equation}
be the canonical projection,  given by $b\mapsto [b]$.

A particularly relevant feature is that $\equiv_a$ is a congruence if $\bbA$ is a Boolean algebra, a Heyting algebra, a bounded distributive lattice or a frame (as stated in Fact \ref{fct:equiva is a congruence} below).
Hence, $\bbA^a$ is canonically endowed with the same algebraic structure of $\bbA$ in each of these cases.
The following properties of $\equiv_a$ are as crucial for the development as they are straightforward:

\begin{fact}
\label{fct: canonical representant of eq class}
Let $\bbA$ be a $\wedge$-semilattice and let $a$ be an action structure over $\bbA$.
\begin{enumerate}
\item $[b\wedge Pre_a] = [b]$ for every $b\in \prod_a\bbA$. Hence, for every $b\in \prod_a\bbA$, there exists a unique  $c\in \prod_a\bbA$ such that\ $c\in [b]_a$ and $c\leq Pre_a$.
\item For all $b, c\in \prod_a\bbA$, we have that $[b]\leq [c]$ iff $b\wedge Pre_a\leq c\wedge Pre_a$.
\item If $\bbA$ is a Heyting algebra, then $[a\to b] = [b]$ for every $b\in \prod_a\bbA$. \end{enumerate} \end{fact}

Item 1 of the fact above implies that each $\equiv_a$-equivalence class has a canonical representant, namely the only element in the given class which is less than or equal to $Pre_a$.
Hence, the map
\begin{equation}\label{eq:def-i'}
i' = i'_{a}: \bbA^a\to \prod_a\bbA
\end{equation} given by $[b]\mapsto b\wedge Pre_a$ is well defined.
Clearly, $\pi\circ i'$ is the identity map on $\bbA^a$.

As was the case in \cite{MPS}, the map $i'$ will be a critical ingredient for the definition of the interpretation of IEAK-formulas on algebraic models (cf.\ Definition \ref{def: extension}). Indeed, 
whenever $\bbA = \f^+$ for some (classical) Kripke frame $\mathcal{F}$, by Proposition \ref{prop: diamond pseudo product okay}, the algebra $\prod_a\bbA$ can be identified with the complex algebra $(\coprod_a\mathcal{F})^+$, and then, by \cite[Fact 9.3]{MPS}, $\bbA^a$ can be identified with ${\mathcal{F}^a}^+$;
then, by \cite[Proposition \ref{prop:i'}]{MPS}, the map $i'$ can be identified with the direct image map of the injection $i: \mathcal{F}^a\to \coprod_a\mathcal{F}$  modulo the isomorphism $\bbA^a \cong {\mathcal{F}^a}^+$. Hence we get the following

\begin{proposition}
\label{prop:i'}
If $\bbA = \mathcal{F}^+$ and $a$ is an action structure over $\bbA$, then $i'(c) = i[\mu(c)]$ for every  $c\in \bbA^a$, where $\mu: \bbA^a \to {\mathcal{F}^a}^+$ is the  BAO-isomorphism identifying the two algebras. Diagrammatically:
\[
\xymatrix@C=20pt@R=30pt{
(\mathcal{F}^+)^a \ar[rr]^\mu\ar[dr]_{i'} && (\mathcal{F}^a)^+ \ar[dl]^{{i[\cdot]}}\\
&\coprod_a\mathcal{F}^+ &
}
 \]
It immediately follows that $i[c] = i'(\nu(c))$ for every  $c\in {\mathcal{F}^a}^+$, where $\nu: {\mathcal{F}^a}^+\to \bbA^a$ is the inverse of $\mu$.
\end{proposition}

\noindent The following compatibility properties of $\equiv_a$ immediately follow from \cite[Fact 7]{MPS} and the general properties of the $|K|$-fold product algebra construction.  
\begin{fact}
\label{fct:equiva is a congruence}
For every $\wedge$-semilattice $\bbA$ and every action structure $a$ over $\bbA$,
\begin{enumerate}
\item the relation $\equiv_a$ is a congruence of $\prod_a\bbA$.
\item If $\bbA$ is a distributive lattice, then $\equiv_a$ is a congruence of $\prod_a\bbA$.
\item If $\bbA$ is a frame, then $\equiv_a$ is a congruence of $\prod_a\bbA$.
\item If $\bbA$ is a Boolean algebra, then  $\equiv_a$ is a congruence of $\prod_a\bbA$.
\item If $\bbA$ is a Heyting algebra, then $\equiv_a$ is a congruence of $\prod_a\bbA$.
\end{enumerate}
\end{fact}

\subsection{Modal operations on the quotient algebra}
\label{subsec: box and diamonds on the pseudo quotient}
 As discussed in \cite[Example 8]{MPS}, the equivalence relation defined in the previous subsection is not in general compatible with the modal operators of the algebra on the domain of  which it is defined. When specialized to the present setting, this implies that $\bbA^a$ does not canonically inherit the structure of modal expansion from $\prod_a\bbA$. In \cite{MPS}, modalities have been defined on the  algebra $\bbA^a$, understood in the general setting, in such a way that, when $\bbA = \mathcal{F}^+$ for some Kripke frame $\mathcal{F}$, it holds that $\bbA^a\cong_{BAO} {\mathcal{F}^a}^+.$ In what follows, we specialize those definitions to the present setting.

For every Heyting algebra $\bbA$, every action structure $a$ over $\bbA$, and every $b\in \prod_a\bbA$, let $$\Diamond^a[b] := [\Diamond^{\prod_a\bbA}(b\wedge Pre_a)\wedge Pre_a] = [\Diamond^{\prod_a\bbA}(b\wedge Pre_a)],$$
$$\Box^a[b] := [Pre_a\to \Box^{\prod_a\bbA}(Pre_a\to b)] = [\Box^{\prod_a\bbA}(Pre_a\to b)].$$
 The right-hand equality in the topmost displayed clause immediately follows from definition, and the one in the displayed clause right above has been justified in \cite[Section 3.2.2]{MPS} in the general setting. The following facts are immediate consequences of Propositions \ref{prop:product modalities normal} and \ref{fct:pseudo products of MHAs}, and of \cite[Facts 9, 10, 11]{MPS}.


\begin{fact}
\label{fct:diamond^a is what we want}
For every HAO $(\bbA, \Diamond)$ and every action structure $a$ over $\bbA$,
\begin{enumerate}
\item $\Diamond^a$ is a normal modal operator. Hence  $(\bbA^a, \Diamond^a)$ is a HAO.
\item If $\bbA = \mathcal{F}^+$ for some Kripke frame $\mathcal{F}$, then $\bbA^a\cong_{BAO} {\mathcal{F}^a}^+.$
\end{enumerate}
\end{fact}

\begin{fact}
\label{fct:box^a is what we want}
For every HAO $(\bbA, \Box)$ and every action structure $a$ over $\bbA$,
\begin{enumerate}
\item $\Box^a$ is a normal modal operator.
\item If $(\bbA, \Box)$ is a BAO and $\Box = \neg\Diamond\neg$, then $\Box^a = \neg\Diamond^a\neg$.
\item If $\bbA = \mathcal{F}^+$ for some Kripke frame $\mathcal{F}$, then $\Box^a = [R^a]$, hence $\bbA^a\cong_{BAO} {\mathcal{F}^a}^+.$
\end{enumerate}
\end{fact}

\begin{fact}
For every HAO $(\bbA, \Diamond, \Box)$ and every action structure $a = (K, k, \alpha, Pre_a)$ over $\bbA$,
\label{fct:pseudo quotients of MHAs}
\begin{enumerate}
\item if $(\bbA, \Diamond, \Box)$ is a MHA and $\alpha$ is an equivalence relation, $(\bbA^a, \Diamond^a, \Box^a)$ is a MHA.
\item If $(\bbA, \Diamond, \Box)$ is a FSA, the algebra $(\bbA^a, \Diamond^a, \Box^a)$ is a FSA.
\item For every tense HAO $(\bbA, \Diamond, \Box, \Diamondblack, \blacksquare)$, the algebra $(\bbA^a, \Diamond^a, \Box^a, \Diamondblack^a, \blacksquare^a)$ is a tense HAO.
\end{enumerate}
\end{fact}

\begin{definition}
\label{def:updated algebra}
For every FSA/MHA $(\bbA, \Diamond, \Box)$ and every action structure $a = (K, k, \alpha, Pre_a)$ over $\bbA$, let $\bbA^a = (\bbA^K/\equiv_a, \Diamond^a, \Box^a)$, defined as above, be the {\em update} of $\bbA$ with $a$.
\end{definition}

\section{Intuitionistic EAK}

\subsection{Axiomatization}
\label{subsec:Axiomatization}
Let \textsf{AtProp} be a countable set of proposition letters. The formulas of the (single-agent) {\em intuitionistic logic of epistemic actions and knowledge} IEAK are built up by the following syntax rule (and let $\mathcal{L}_{IEAK}$ denote the resulting set of formulas):

\begin{center}
$\phi::= p\in \mathsf{AtProp} \mid \bot \mid  \phi\vee \phi\mid \phi\wedge \phi\mid \phi\to \phi \mid  \Diamond\phi \mid  \Box\phi\mid 
\langle\alpha\rangle \phi\mid [\alpha]\phi\;\; (\alpha\in \mathsf{Act(\mathcal{L})}).$ 
\end{center}
The same stipulations hold for the defined connectives $\top$, $\neg$ and $\leftrightarrow$ as introduced early on.
IEAK is axiomatically defined by the axioms and rules of IK (MIPC) plus the following axioms:

\noindent
\begin{center}
\begin{tabular}{l l}
{\bf Interaction with logical constants} & {\bf Preservation of facts}\\
$\langle \alpha \rangle \bot \leftrightarrow \bot$, $\ \langle \alpha \rangle \top \leftrightarrow Pre(\alpha)$  & $\langle \alpha \rangle p \leftrightarrow Pre(\alpha) \wedge p$\\
$[\alpha] \top \leftrightarrow \top$, $\ [\alpha] \bot \leftrightarrow \neg Pre(\alpha)$ & $[\alpha]p \leftrightarrow Pre(\alpha) \rightarrow p$ \\

{\bf Interaction with disjunction} & {\bf Interaction with conjunction}\\

$\langle \alpha \rangle (\phi \vee \psi) \leftrightarrow \langle \alpha \rangle \phi \vee \langle \alpha \rangle \psi$ &  $\langle \alpha \rangle (\phi \wedge \psi) \leftrightarrow \langle \alpha \rangle \phi \wedge \langle \alpha \rangle \psi$\\

$[\alpha] (\phi \vee \psi) \leftrightarrow Pre(\alpha) \rightarrow (\langle \alpha \rangle \phi \vee \langle \alpha \rangle \psi)$ &  $[\alpha](\phi \wedge \psi) \leftrightarrow [\alpha]\phi \wedge [\alpha] \psi$ \\

{\bf Interaction with implication}& \\
 $\langle \alpha \rangle (\phi \rightarrow \psi) \leftrightarrow Pre(\alpha)\wedge  (\langle \alpha \rangle \phi\rightarrow \langle \alpha\rangle \psi)$ & \\

 $[\alpha] (\phi \rightarrow \psi) \leftrightarrow \langle \alpha \rangle \phi\rightarrow \langle \alpha\rangle \psi$ & \\

{\bf Interaction with diamond} & {\bf Interaction with box}\\
 $\langle\alpha\rangle\Diamond\phi \leftrightarrow Pre(\alpha)\wedge\bigvee\{\Diamond\langle\alpha_j\rangle \phi\mid k\alpha j\}$ & $\langle\alpha\rangle\Box\phi \leftrightarrow Pre(\alpha)\wedge\bigwedge\{\Box[\alpha_j] \phi\mid k\alpha j\}$\\
 $[\alpha]\Diamond\phi \leftrightarrow Pre(\alpha)\rightarrow\bigvee\{\Diamond\langle\alpha_j\rangle \phi\mid k\alpha j\}$ & $[\alpha]\Box\phi \leftrightarrow Pre(\alpha)\rightarrow\bigwedge\{\Box[\alpha_j] \phi\mid k\alpha j\}$
\end{tabular}
\end{center}
where, for every action structure $\alpha = (K, k, \alpha, Pre_\alpha)$, and every $j\in K$, the action structure $\alpha_j$ is defined as $\alpha_j = (K, j, \alpha, Pre_\alpha)$.

\subsection{Models}
\label{ssec:models}
\begin{definition}
\label{def:update model}
An {\em algebraic model} is a tuple $M = (\bbA, V)$ such that\ $\bbA$
is an FSA (resp. an MHA) (cf.\ Definition \ref{def: IEA}) and $V:
\mathsf{AtProp}\to \bbA$.
For every algebraic model $M$ and every action structure $\alpha$ over
$\mathcal{L}$, let
$$\prod_\alpha M: = (\prod_\alpha\bbA, \prod_\alpha
V)$$
where $\prod_\alpha\bbA: = \prod_a\bbA$, and $a$ is the action
structure over $\bbA$ induced by $\alpha$ via $V$ (cf.\ introduction
of Section \ref{sec:EUA}); moreover, $(\prod_\alpha V) (p) := \prod_a
V(p)$ for every $p\in \mathsf{AtProp}$. Likewise, we can define
$$M^\alpha: = (\bbA^\alpha, V^\alpha)$$
where $\bbA^\alpha: = \bbA^a$ (cf.\ Definition \ref{def:updated
  algebra}), and $V^\alpha := \pi \circ \prod_\alpha V$ (cf.\
\eqref{eq:def-pi}).
\end{definition}

Given an algebraic model $M = (\bbA, V)$, we want to define its
associated extension map $\val{\cdot}_M: \mathcal{L}_{IEAK} \to \bbA$
so that, when $\bbA = \mathcal{F}^+$ for some Kripke frame
$\mathcal{F}$, we recover the familiar extension map associated with
the model $M = (\mathcal{F}, V)$.
To this end, we introduce the notation
\begin{equation}\label{eq:iik}
\xymatrix{
M \ar[r]^{\iota_k} &\coprod_\alpha M & \ar[l]_{i} M^\alpha
}
\end{equation}
where the map $i: M^\alpha\to \coprod_\alpha M$ is the submodel
embedding, and $\iota_k: M\to \coprod_\alpha M$ is the embedding of
$M$ into its $k$-colored copy, which, by convention, is the copy
corresponding to the distinguished point of $\alpha$.

Notice that -- when $M$ is a
relational model -- the satisfaction condition for
$\langle\alpha\rangle$-formulas
$$M,
w\Vdash \langle \alpha \rangle \phi\quad \mbox{ iff } \quad M, w\Vdash
Pre(\alpha) \mbox{ and } M^\alpha, (w, k)\Vdash \phi$$
can be equivalently written as follows:
$$w\in \val{ \langle \alpha \rangle \phi}_M\quad \mbox{ iff } \quad
\exists x\in W^\alpha \mbox{ such that }\ x\in \val{\phi}_{M^\alpha}\
\mbox{ and }\ i(x)= \iota_k(w) \in\val{Pre(\alpha)}_{\coprod_\alpha
  M}, $$
%
Because $i$ is injective, we get that
$x\in \val{\phi}_{M^\alpha}$ iff $\iota_k(w) = i(x)\in
i[\val{\phi}_{M^\alpha}]$, iff $w\in \iota_k^{-1}
[i[\val{\phi}_{M^\alpha}]]$. Hence,
$$w\in \val{ \langle \alpha
  \rangle \phi}_M\quad \mbox{ iff } \quad w\in\val{Pre(\alpha)}_M \cap
\iota_k^{-1} [i[\val{\phi}_{M^\alpha}]],$$
from which we get that
\begin{equation}\label{eq:extension alpha phi}\val{ \langle \alpha \rangle \phi}_M = \val{Pre(\alpha)}_M \cap \iota_k^{-1} [i[\val{\phi}_{M^\alpha}]].\end{equation}
Likewise, equivalently rewriting the following satisfaction condition for $[\alpha]$-formulas $$M, w\Vdash  [\alpha] \phi\quad \mbox{ iff } \quad  M, w\Vdash  Pre(\alpha)  \mbox{ implies } M^\alpha, (w, k)\Vdash \phi$$
yields:
%
\begin{equation}\label{eq:extension box alpha phi}\val{ [\alpha]
    \phi}_M = \val{Pre(\alpha)}_M \Rightarrow \iota_k^{-1}
  [i[\val{\phi}_{M^\alpha}]],
\end{equation}
where $X\Rightarrow Y = (W\setminus X)\cup Y$ for every $X, Y\subseteq
W$.
%
To see that \eqref{eq:extension box alpha phi} is `in algebraic form',
recall that the dual of
\eqref{eq:iik} is written as
\begin{equation}\label{eq:i'pik}
\xymatrix{
\bbA  &\ar[l]^{\pi_k}\prod_\alpha \bbA \ar[r]_{\pi}& \ar@/_15pt/[l]_{i'} \bbA^\alpha
}
\end{equation}
where $\pi_k$ is the projection onto the $k$-th coordinate and $\pi$
and $i'$ are as in \eqref{eq:def-pi} and \eqref{eq:def-i'}, with $i'$
being left-adjoint to $\pi$. To say
that \eqref{eq:i'pik} is the dual of \eqref{eq:iik} means precisely
that in the case of  $\bbA={\mathcal{F}^{\alpha}}^+$ we have
$\pi_k=i_k^{-1}$ and $\pi=i^{-1}$ and $i'=i[-]$, see Proposition \ref{prop:i'}.
So we can adopt equations (\ref{eq:extension alpha phi}) and
(\ref{eq:extension box alpha phi})---modified by replacing
$i[\cdot]$ and $\iota_k$ with $i'$ and
$\pi_k$---in {\em any} algebraic model
$(\bbA, V)$:

\begin{definition}
\label{def: extension}
For every algebraic model $M = (\bbA, V)$, the {\em extension map} $\val{\cdot}_M: \mathcal{L}_{IEAK}\to \bbA$ is defined recursively as follows:
\begin{center}
\begin{tabular}{r c l}
$\val{p}_M$ &$=$& $V(p)$\\
$\val{\bot}_M$&$ =$&$ \bot^\bbA$\\
$\val{\phi\vee \psi}_M$ &$=$& $\val{\phi}_M\vee^\bbA\val{\psi}_M$\\
$\val{\phi\wedge \psi}_M$ &$=$& $\val{\phi}_M\wedge^\bbA\val{\psi}_M$\\
$\val{\phi\to \psi}_M$ &$=$& $\val{\phi}_M\rightarrow^\bbA\val{\psi}_M$\\
$\val{\Diamond\phi}_M$ &$=$& $\Diamond^\bbA\val{\phi}_M$\\
$\val{\Box\phi}_M$ &$=$& $\Box^\bbA\val{\phi}_M$\\
$\val{\langle \alpha \rangle \phi}_M$ &$ = $&$\val{Pre(\alpha)}_M \wedge^\bbA  \pi_k\circ i'(\val{\phi}_{M^\alpha})$\\
$\val{ [\alpha] \phi}_M$ &$ = $&$\val{Pre(\alpha)}_M \rightarrow^\bbA  \pi_k\circ i'(\val{\phi}_{M^\alpha}).$\\
\end{tabular}
\end{center}
\end{definition}\medskip

 Notice that, by Proposition \ref{prop:algebras and frames}, the above definition specializes to those algebraic models $(\bbA, V)$ such that $\bbA = \f^+$ is the complex algebra of some IK-frame (MIPC-frame) $\f$, and from those, to their relational counterparts $(\f, V)$. Hence, as a special case of the definition above we get an interpretation of IEAK on relational IK-models (MIPC-models). More details about these models are reported in  the next subsection.
\subsection{Relational semantics for IEAK}
\label{subsec:rel semantics}
In order to recover the relational semantics of IEAK from its more general semantics given by the algebraic models of Definition \ref{def: extension}, we need to dually characterize back the FSAs (MHAs) and the update construction from  $\bbA$ to $\bbA^a.$ As is well known (cf.\ e.g.\ \cite{Bez98, Bez99}), dually characterizing the FSAs (MHAs) is possible in full generality, and the resulting construction involves the intuitionistic counterparts of descriptive general frames in classical modal logic, i.e.\ relational structures endowed with topologies. However, obtaining the purely relational IK-frames (MIPC-frames) is possible for certain special FSAs (MHAs), which we call {\em perfect} FSAs (MHAs). This dual characterization has been reported on in detail in \cite[Section 4.3]{MPS}, where the update construction on intuitionistic relational models has been also spelled out in the special case of public announcements. In what follows, we provide the relevant definitions and facts to perform the dual characterization in the case of updates by means of general action structures, omitting proofs whenever they already appear in \cite{MPS}, and including  proofs whenever they do not appear anywhere to the authors' knowledge.

\bigskip For every poset $P = (X,\leq)$, a non-bottom element $x\in X$ is {\em completely join-prime} if, for every $S\subseteq X$ such that $x \leq \bigvee S$, there exists some $s\in S$ such that $x\leq s$; a non-top element $y\in X$ is {\em completely meet-prime} if, for every $S\subseteq X$ such that $\bigwedge S\leq y$, there exists some $s\in S$ such that $s\leq y$. Let $\jty(P)$ and $\mty(P)$ respectively denote the set of the completely join-prime elements and the set of the completely meet-prime elements in $P$. A poset $P$ is a {\em complete lattice} if the joins and meets of arbitrary subsets of $P$ exist, in which case, $P$ is {\em completely distributive} if arbitrary meets distribute over arbitrary joins.
$P$ is {\em completely join-generated} (resp.\ {\em completely meet-generated}) by a given $S\subseteq P$ if for every $x\in P$, $x = \bigvee S'$ (resp.\ $x = \bigwedge S'$) for some $S'\subseteq S$.
\begin{definition}
An HA $\bbA$ is {\em perfect} if it is a complete and completely distributive lattice w.r.t.\ its natural ordering, and is also completely join-generated by $\jty(\bbA)$ (or equivalently, completely meet-generated by $\mty(\bbA)$). An HAO $(\bbA, \Diamond, \Box)$ is {\em perfect} if $\bbA$ is a perfect HA, and moreover, $\Diamond$ distributes over arbitrary joins and $\Box$ distributes over arbitrary meets. A {\em perfect} FSA (MHA) is an FSA (MHA) which is also a perfect HAO.
\end{definition}
Clearly, any finite HA(O) is perfect. It is well known that a Heyting algebra $\bbA$ is perfect iff it is isomorphic to $\mathcal{P}^{\downarrow}(P)$, where $P = (\jty(\bbA), \leq)$ and $\leq$ is the restriction of the natural ordering of $\bbA$ to $\jty(\bbA)$. The Boolean self-duality $u\mapsto \neg u$ generalizes, in the HA setting, to the maps $\kappa: \bbA\to \bbA$, given by $x\mapsto \bigvee\{x'\mid x'\nleq x\}$, and
$\lambda: \bbA\to \bbA$, given by $y\mapsto \bigwedge\{y'\mid y\nleq y'\}$. These  maps induce  order isomorphisms $\kappa: \jty(\bbA)\to \mty(\bbA)$ and $\lambda: \mty(\bbA)\to \jty(\bbA)$ (seen as subposets of $\bbA$). Clearly, $x\nleq \kappa(x)$ (resp.\ $\lambda(y)\nleq y$) for every $x\in \jty(\bbA)$ (resp.\ $y\in \mty(\bbA)$); moreover, for every $u\in \bbA$ and every $x\in \jty(\bbA)$, $$j\leq u\quad \mbox{ iff }\quad u\nleq\kappa(j).$$
By the theory of adjunction on posets, it is well known  that, in a perfect HAO $\bbA$, the properties of complete distributivity enjoyed by the  modal operations imply that they are parts of adjoint pairs:  unary operations $\Diamondblack$ and $\blacksquare$ are defined on $\bbA$ so that for all $x, y\in \bbA$, $$\Diamond x\leq y\quad \mbox{ iff }\quad x\leq \blacksquare y\quad\quad\mbox{ and }\quad\quad \Diamondblack x\leq y\quad \mbox{ iff }\quad x\leq \Box y.$$
We denote these adjunction relations by writing $\Diamond\dashv \blacksquare$ and $\Diamondblack\dashv \Box$. One member of the adjunction relation completely determines the other. The choice of notation is a reminder of the fact that, by the general theory, $\Diamondblack$ distributes over arbitrary joins (i.e., it enjoys exactly the characterizing property of a `diamond' operator on perfect algebras), and $\blacksquare$ distributes over arbitrary meets (i.e., it enjoys the characterizing property of a `box' operator on perfect algebras). In particular, they are  both order-preserving. Well known pairs of adjoint modal operators occur in temporal logic: its axiomatization essentially states that, when interpreted on algebras, the forward-looking diamond is {\em left adjoint} to the backward-looking box, and the backward-looking diamond is  left adjoint to the forward-looking box. This is actually an essential feature: indeed $R$ is the accessibility relation for one operation iff $R^{-1}$ is the accessibility relation for the other.

Let us now introduce the intuitionistic counterpart of the atom structures for complete atomic BAOs:
\begin{definition}
\label{def: atom structure}
For every perfect FSA (MHA) $\bbA$, let us define $R\subseteq \jty(\bbA)\times\jty(\bbA)$ by setting $$xRy\quad \mbox{ iff }\quad x\leq \Diamond y\ \mbox{ and }\ y\leq \Diamondblack x.$$
The {\em prime structure} associated with $\bbA$ is the relational structure $\bbA_+: = (\jty(\bbA), \leq, R).$
\end{definition}
Notice that $y\leq \Diamondblack x\quad \mbox{ iff }\quad \Diamondblack x \nleq \kappa(y)\quad \mbox{ iff }\quad x\nleq \Box\kappa(y).$
\begin{fact} For every perfect HAO $\bbA$,
\begin{enumerate}
\item if  $\bbA$ is an FSA, then $\bbA_+$ is an IK-frame;

\item if  $\bbA$ is an MHA, then  $\bbA_+$ is an MIPC-frame.
\end{enumerate}
\end{fact}

\begin{proposition}
For every perfect FSA $\bbA$, and every IK-frame $\f$,  $$\bbA\cong_{HAO} (\bbA_+)^+\quad \mbox{ and }\quad\f\cong (\f^+)_+.$$
\end{proposition}

The bijective correspondence above, between perfect FSAs and IK-frames, specializes to MHAs and MIPC-frames, and also extends to  homomorphisms and p-morphisms; in short, it is a duality, but treating it in detail is out of the aims of the present paper.

\begin{definition}
For every IK-frame $\f = (W, \leq, R)$ and every action structure  $a = (K, k, \alpha, Pre_a)$ over  the complex algebra $\f^+$ , let $\f^a = (W^a, \leq^a, R^a)$ be defined in the usual way, i.e., as the subframe of the intermediate structure $\coprod_a\f: = (W\times K, R\times \alpha)$ determined by the subset $$W^a: = \{(w, j)\in W\times K\mid w\in Pre_a(j)\}.$$
\end{definition}
Because $Pre_a(j)$ is a down-set for every $j\in K$, it is easy to see that $\f$ being an IK-frame implies that $\f^a$ is an IK-frame, and that the analogous result holds w.r.t.\ MIPC-frames if $\alpha$ is an equivalence relation. The remainder of the present subsection focuses on showing that, for every perfect FSA $\bbA$ and every action structure $a$ over $\bbA$, $$(\bbA^a)_+ \cong(\bbA_+)^a.$$

\begin{fact}
For every HA $\bbA$ and every action structure $a = (K, k, \alpha, Pre_a)$ over $\bbA$,
\begin{enumerate}
\item the set $\jty(\prod_a\bbA)$ bijectively corresponds to $\coprod_a\jty(\bbA)\cong \jty(\bbA)\times K$.
\item The accessibility relation $R^{\prod_a}$ of the prime structure $(\prod_a\bbA)_+$ bijectively corresponds to the product relation $R\times \alpha$ (where $R$ is the relation of the prime structure $\bbA_+$) under the identification of item 1 above.
    \item $(\prod_a\bbA)_+\cong \coprod_a\bbA_+$.
\end{enumerate}
\end{fact}
\begin{proof}
1. It is enough to show that $b: K\to \bbA\in \jty(\prod_a\bbA)$ iff   there exists a unique $j\in K$ such that $b(j)\in \jty(\bbA)$, and $b(i) = \bot$ for $i\in K\setminus\{j\}$. The direction from right to left is clear. Conversely, if $b\in \jty(\prod_a\bbA)$ and $j\in K$ such that $b(j)\neq \bot$, then $b(j)\in \jty(\bbA)$; indeed, for every $S\subseteq \bbA$ such that $b(j)\leq \bigvee S$, consider the collection $S'\in \prod_a\bbA$ whose elements are the maps $c: K\to \bbA$ such that $c(j)\in S$ and $c(i) = \top$ for $i\neq j$. To finish the proof, if $b(i)\neq \bot$ for more than one $i\in K$, then $b\leq \bigvee_{j\in K} c_j$, where for every $j\in K$, the map $c_j:K\to \bbA$  sends $j$ to $b(j)$ and every other element of $K$ to $\bot$, but $b\nleq c_j$ for any $j\in K$.\\
2. Fix $b, c\in \jty(\prod_a\bbA)$. By the statement proved in item 1 above, $b$ and $c$ can be respectively identified with $(b(i), i), (c(j), j)\in \jty(\bbA)\times K$ for some unique $i, j\in K$, so that for every $i\in K$, $$(\Diamond^{\prod_a\bbA} c)(i) = \bigvee\{\Diamond^\bbA c(i')\mid i\alpha i'\}= \begin{cases}
\Diamond^\bbA c(j) & \mbox{ if } i\alpha j\\
\bot & \mbox{ otherwise,}\\
\end{cases}$$
and for every $j\in K$, $$(\Diamondblack^{\prod_a\bbA} b)(j) = \bigvee\{\Diamondblack^\bbA b(i')\mid i'\alpha j\}= \begin{cases}
\Diamondblack^\bbA b(i) & \mbox{ if } i\alpha j\\
\bot & \mbox{ otherwise.}\\
\end{cases}$$
Hence, we have:
\begin{center}
\begin{tabular}{r c l l}
$bR^{\prod_a} c$ & iff & $b\leq \Diamond^{\prod_a\bbA} c$ and $c\leq \Diamondblack^{\prod_a\bbA} b$ &\\
& iff & $b(i)\leq (\Diamond^{\prod_a\bbA} c)(i)$ and $c(j)\leq (\Diamondblack^{\prod_a\bbA} b)(j)$ &\\

& iff & $i\alpha j$ and $b(i)\leq \Diamond^\bbA c(j)$, and $i\alpha j$ and $c(j)\leq \Diamondblack^{\bbA} b(i)$ &\\
& iff & $i\alpha j$ and $b(i) R c(j)$ &\\
& iff & $(b(i),i)(R\times \alpha)(c(j), j).$ &\\
\end{tabular}
\end{center}
3. From the previous items it immediately follows that both the universes and the accessibility relations of the  structures $(\prod_a\bbA)_+$ and $\coprod_a\bbA_+$ can be identified. It remains to be shown that their ordering relations can be identified too. Indeed, if $b, c: K\to \bbA\in \jty(\prod_a\bbA)$ are respectively identified with $(b(i), i), (c(j), j)\in \jty(\bbA)\times K$ for some unique $i, j\in K$, then $b\leq_{(\prod_a\bbA)_+} c$ iff $b(i')\leq c(i')$ for every $i'\in K$, iff  $i = j$ and $b(i)\leq c(j)$, iff $(b(i), i)\leq_{\coprod_a\bbA_+} (c(j), j)$.
\end{proof}

Fact 19 in \cite{MPS} (and the discussion below it), when specialized to the present setting, states that the prime structure of the  quotient of $\prod_a\bbA$ by means of $\equiv_a$ is identifiable with the subframe of $\coprod_a\bbA_+$ determined by the subset $\{(x, j)\in \jty(\bbA)\times K\mid x\in Pre_a(j)\}$.
This, together with the fact above, readily imply that $(\bbA^a)_+ \cong(\bbA_+)^a.$

The identification between these two relational structures implies that the mechanism of epistemic update remains completely unchanged when generalizing from the Boolean to the intuitionistic setting.

\subsection{Soundness and completeness for IEAK}
\begin{proposition}
\label{prop:IPAL soundness}
IEAK is sound with respect to algebraic IK-models (MIPC-models), hence with respect to relational IK- models (MIPC-models).
\end{proposition}
\begin{proof}
The soundness of the preservation of facts and logical constants follows from Lemma \ref{fct: preservation of facts}. The soundness of the remaining axioms is proved in Lemmas
\ref{fct:normality of dynamic modalities}, \ref{fct:uncongenial distribution over lattice struct}, \ref{fct:distribution over implication}, \ref{fct:distribution over diamond}, \ref{fct:distribution over box} of the appendix. 
\end{proof}

\begin{theorem}
\label{th:completeness}
IEAK is complete with respect to relational IK-models (MIPC-models).
\end{theorem}
\begin{proof}
The proof is analogous to the proof of completeness of classical EAK~\cite[Theorem 3.5]{BMS}, and follows from the reducibility of IEAK to IK (MIPC) via the reduction axioms. Let $\phi$ be a valid IEAK formula. Let us consider some innermost occurrence of a dynamic modality in $\phi$. Hence, the subformula $\psi$ having that occurrence labeling the root of its generation tree is either of the form $[\alpha] \psi'$ or of the form $\langle \alpha \rangle \psi'$, for some formula $\psi'$ in the static language.
The distribution axioms make it possible to equivalently transform $\psi$ by pushing the dynamic modality down the generation tree, through the static connectives, until it attaches to a proposition letter or to a constant symbol. Here, the dynamic modality disappears, thanks to an application of the appropriate `preservation of facts' or `interaction with logical constant'  axiom. This process is repeated for all the dynamic modalities of $\phi$, so as to obtain a formula $\phi'$ which is provably equivalent to $\phi$. Since $\phi$ is valid by assumption, and since the process preserves provable equivalence, by soundness we can conclude that $\phi'$ is valid. By  Proposition \ref{MIPC s&c}, we can conclude that $\phi'$ is provable in IK (MIPC), hence in IEAK. This, together with the provable equivalence of $\phi$ and $\phi'$, concludes the proof.
\end{proof}

\section{An illustration}\label{sec:illustration}
Let us recall from Example~\ref{exle:actions} the following scenario. There is a set $I$ of three agents, $\asf, \bsf, \csf$, and three cards, two of which are white, and are each held by $\bsf$ and $\csf$, and one is green, and is held by $\asf$. Initially, each agent only knows the color of its own card, and it is common knowledge among the three agents that there are two white cards and one green one. Then $\asf$ shows its card only to $\bsf$, but in the presence of $\csf$. Then $\bsf$ announces that $\asf$ knows what the actual distribution of cards is. Then, after having witnessed $\asf$ showing its card to $\bsf$, and after the ensuing public announcement of $\bsf$, agent $\csf$ knows what the actual distribution is.

This scenario is less of a puzzle than the Muddy Children, but it illustrates an action more complicated than a public announcement. In both scenarios, a given subgroup of agents draws conclusions on factual states of affairs purely based, besides the initial information, on information about other agents' epistemic states.

The purpose of this section is to illustrate that reasoning such as this can be supported on an intuitionistic base by IEAK. Of course, we will need the appropriate multi-agent version of it, which we denote IEAK$_I$, whose language, if the set of agents is taken to be $I = \{\asf, \bsf, \csf\}$, is defined as one expects by considering indexed epistemic modalities $\Box_\isf$ and $\Diamond_\isf$ for $ \isf\in I$, and whose axiomatization is given by correspondingly indexed copies of the IEAK axioms\footnote{For the remainder of this section, if $L$ is one of the logics introduced so far, $L_I$ will denote its indexed version. For any logic $L$, the relation of provable equivalence relative to $L$ will be denoted by $\dashv\vdash_{L}$.}.
For the sake of this scenario, we can restrict the set of proposition letters to $\{W_\isf, G_\isf\mid \isf\in I\}$. The intended meaning of $W_\isf$ and $G_\isf$ is `agent $\isf$ holds a white card', and `agent $\isf$ holds a green card' respectively.

Derived modalities can be defined in the language of IEAK$_I$, which will act as finitary approximations of common knowledge: for every IEAK$_I$-formula $\phi$, let $E\phi = \bigwedge_{\isf \in I}\Box_\isf\phi$. The intended meaning of $E$ is `Everybody knows'. It is easy to see that $E\top \dashv\vdash_{IK_I} \top$ and $E(\phi\wedge \psi) \dashv\vdash_{IK_I} E\phi\wedge E\psi$. So $E$ is a box-type normal modality. 

The action structure $\alpha$ encoding the action performed by agent
$\asf$ can be assimilated to the atomic proposition $G_\asf$ being announced
to the subgroup $\{\asf, \bsf\}$. Hence, $\alpha = (K, k, \alpha_\asf,
\alpha_\bsf, \alpha_\csf, Pre_\alpha)$ can be specified as follows: $K =
\{k, l\}$; moreover, $Pre(\alpha) = Pre_\alpha(k) = G_\asf$, and
$Pre(\alpha_l) = Pre_\alpha(l) = W_\asf$; finally, $\alpha_\asf = \alpha_\bsf =
\Delta_K$ and $\alpha_\csf = K\times K$.

The action structure $\beta$ encoding the public announcement
performed by agent $\bsf$ can be specified as a one-state structure, the
precondition of which is the formula $Pre (\beta) = \bigwedge_{i \in
  I}(G_\isf\rightarrow \Box_\asf G_\isf)$.

Let us introduce the following abbreviations:
\begin{itemize}
\item $\mathsf{aut}:= \bigwedge_{\isf \in I} [(W_\isf\to
  \bot)\leftrightarrow G_\isf]$ expresses the fact that holding a white
  or a green card are both {\em mutually incompatible} and {\em
    exhaustive} conditions;
\item $\mathsf{one}:= \bigvee_{\isf\in I} (G_\isf\wedge\bigwedge_{\hsf \neq \isf}
  W_\hsf)$ expresses the fact there are two white cards and one green
  one;
\item $\mathsf{other?} :=\bigwedge_{\isf\in I} (W_\isf\rightarrow
  \bigwedge_{\hsf\neq \isf}\Diamond_\isf G_\hsf)$ expresses the fact that any
  agent holding a white card does not know who of the other two agents
  holds the green card.
\end{itemize}
The aim of this section is proving the following

\begin{proposition}
  Let $\mathcal{L}$ be an extension of IEAK$_I$ with $\mathsf{aut}$
  and $\mathsf{one}$.  Then,
$$ E(\mathsf{other?})\vdash_{\mathcal{L}}
[\alpha][\beta]\Box_\csf G_\asf.$$
\end{proposition}

\begin{proof} The following chain of provable equivalences holds in
  IEAK$_I$:
\begin{center}
\begin{tabular}{c l l}
& $[\alpha][\beta]\Box_\csf G_\asf$ &\\
$\dashv\vdash_{IEAK_I}$& $[\alpha](Pre(\beta)\rightarrow \Box_\csf(Pre(\beta)\rightarrow G_\asf))$ &\\
$\dashv\vdash_{IEAK_I}$& $\langle\alpha\rangle Pre(\beta)\rightarrow \langle\alpha\rangle\Box_\csf(Pre(\beta)\rightarrow G_\asf)$&\\
$\dashv\vdash_{IEAK_I}$& $\langle\alpha\rangle Pre(\beta)\rightarrow (Pre(\alpha)\wedge (\Box_\csf[\alpha](Pre(\beta)\rightarrow G_\asf)\wedge \Box_\csf[\alpha_l](Pre(\beta)\rightarrow G_\asf)))$ &\\
$\dashv\vdash_{IEAK_I}$& $\langle\alpha\rangle Pre(\beta)\rightarrow (Pre(\alpha)\wedge (\Box_\csf(\langle\alpha\rangle Pre(\beta)\rightarrow \langle\alpha\rangle G_\asf)\wedge \Box_\csf( \langle\alpha_l\rangle Pre(\beta)\rightarrow \langle\alpha_l\rangle G_\asf)))$ &\\
$\dashv\vdash_{IEAK_I}$& $ [\langle\alpha\rangle Pre(\beta)\rightarrow Pre(\alpha)]\wedge[\langle\alpha\rangle Pre(\beta)\rightarrow (\Box_\csf(\langle\alpha\rangle Pre(\beta)\rightarrow \langle\alpha\rangle G_\asf)]$&\\
& $\phantom{[\langle\alpha\rangle Pre(\beta)\rightarrow Pre(\alpha)]}\wedge [\langle\alpha\rangle Pre(\beta)\rightarrow\Box_\csf( \langle\alpha_l\rangle Pre(\beta)\rightarrow \langle\alpha_l\rangle G_\asf)]$. &\\

\end{tabular}
\end{center}
Hence, by the Deduction Theorem, it is enough to show that
\begin{center}
\begin{tabular}{r l l r}
$\langle\alpha\rangle Pre(\beta)$ &$\vdash_{IEAK_I}$& $Pre(\alpha)$ & $\quad\quad \quad$ (1)\\
&&\\
$\langle\alpha\rangle Pre(\beta)$ &$\vdash_{\mathcal{L}}$& $\Box_\csf(\langle\alpha\rangle Pre(\beta)\rightarrow \langle\alpha\rangle G_\asf)$ &$\quad\quad \quad$ (2)\\
&&\\
$ E(\mathsf{other?})$ &$\vdash_{\mathcal{L}}$& $\Box_\csf( \langle\alpha_l\rangle Pre(\beta)\rightarrow \langle\alpha_l\rangle G_\asf).$ & $\quad\quad \quad$ (3)\\
\end{tabular}
\end{center}
The entailment (1) straightforwardly follows from the IEAK rewriting axioms, and this verification is left to the reader. As to the remaining ones, notice preliminarily that, because of $\mathsf{aut}$ and $\mathsf{one}$, it holds that $ (G_\hsf\wedge G_\asf)\dashv\vdash_{\mathcal{L}} \bot$ for each $\hsf\in I\setminus\{\asf\}$,
which justifies  the step marked with ($\ast$) in the following chain of provable equivalences:
\begin{center}
\begin{tabular}{r l l}
$\langle\alpha\rangle Pre(\beta)$ &$\dashv\vdash_{IEAK_I}$& $Pre(\alpha)\wedge \bigwedge_{\isf \in I}(\langle\alpha\rangle G_\isf\rightarrow \langle\alpha\rangle \Box_\asf G_\isf)$\\
&$\dashv\vdash_{IEAK_I}$& $Pre(\alpha)\wedge \bigwedge_{\isf \in I} ((G_\isf\wedge G_\asf)\rightarrow \langle\alpha\rangle \Box_\asf G_\isf)$\\
($\ast$)&$\dashv \vdash_{\mathcal{L}}$& $Pre(\alpha)\wedge (G_\asf\rightarrow \langle\alpha\rangle \Box_\asf G_\asf)$\\

&$\dashv\vdash_{IEAK_I}$& $Pre(\alpha)\wedge (G_\asf\rightarrow (Pre(\alpha)\wedge \Box_\asf[\alpha]G_\asf))$\\
&$\dashv\vdash_{IEAK_I}$& $Pre(\alpha)\wedge (G_\asf\rightarrow Pre(\alpha))\wedge (G_\asf\rightarrow \Box_\asf[\alpha]G_\asf)$\\

&$\dashv\vdash_{IEAK_I}$& $G_\asf \wedge (G_\asf\rightarrow \Box_\asf(G_\asf\rightarrow G_\asf))$\\
&$\dashv\vdash_{IEAK_I}$& $ G_\asf.$\\
\end{tabular}
\end{center}
Hence, proving the entailment (2) is equivalent to showing that $G_\asf\vdash_{\mathcal{L}}\Box_\csf(G_\asf\rightarrow G_\asf)$, which is immediate. As to the entailment (3), by the axiom FS2 and the Deduction Theorem, it is enough to show that

\begin{center}
\begin{tabular}{r c l r}
$ E(\mathsf{other?}), \Diamond_\csf\langle\alpha_l\rangle Pre(\beta)$ &$\vdash_{\mathcal{L}}$& $\Box_\csf \langle\alpha_l\rangle G_\asf$. & (4)\\
\end{tabular}
\end{center}
Notice preliminarily that $\mathsf{aut}$ and $\mathsf{one}$ imply that $ (W_\isf\wedge G_\isf)\dashv\vdash_{\mathcal{L}} \bot$ for each $\isf\in I$
(which justifies the equivalence marked with ($\ast$) below),  and also that
$ (G_\isf\wedge \bigwedge_{h\neq i} W_\hsf)\dashv\vdash_{\mathcal{L}} G_\isf$  
for each $\isf\in I$ (which justifies the equivalence marked with ($\ast\ast$) below). 
Hence:

\begin{center}
\begin{tabular}{r l l}
$\langle\alpha_l\rangle Pre(\beta)$ &$\dashv\vdash_{IEAK_I}$& $Pre(\alpha_l)\wedge \bigwedge_{\isf \in I}(\langle\alpha_l\rangle G_\isf\rightarrow \langle\alpha_l\rangle \Box_\asf G_\isf)$\\
($\ast$)&$\dashv\vdash_{\mathcal{L}}$& $W_\asf\wedge [((W_\asf\wedge G_\bsf)\rightarrow \langle\alpha_l\rangle \Box_\asf G_\bsf)\wedge ((W_\asf\wedge G_\csf)\rightarrow \langle\alpha_l\rangle \Box_\asf G_\csf)]$\\
&$\dashv\vdash_{IEAK_I}$& $W_\asf\wedge [((W_\asf\wedge G_\bsf)\rightarrow \Box_\asf[\alpha_l] G_\bsf)\wedge ((W_\asf\wedge G_\csf)\rightarrow \Box_\asf [\alpha_l] G_\csf)]$\\
($\ast\ast$)&$\dashv\vdash_{\mathcal{L}}$& $W_\asf\wedge [(G_\bsf\rightarrow \Box_\asf(W_\asf\rightarrow G_\bsf))\wedge (G_\csf\rightarrow \Box_\asf (W_\asf\rightarrow G_\csf))]$\\

\end{tabular}
\end{center}
Therefore, since $ E(\mathsf{other?})\vdash_{IEAK_I} \Box_\csf(W_\asf\rightarrow (\Diamond_\asf G_\bsf\wedge \Diamond_\asf G_\csf))$, to prove (4) it is enough  to show that
\begin{center}
\begin{tabular}{r c l  }
$\Box_\csf(W_\asf\rightarrow (\Diamond_\asf G_\bsf\wedge \Diamond_\asf G_\csf)), \Diamond_\csf [W_\asf\wedge [(G_\bsf\rightarrow \Box_\asf (W_\asf\rightarrow G_\bsf))\wedge (G_\csf\rightarrow \Box_\asf (W_\asf\rightarrow G_\csf))]]$ &$\vdash_{\mathcal{L}}$& $\bot$. \\
\end{tabular}
\end{center}
To this aim, observe preliminarily that
\begin{center}
\begin{tabular}{r l l l}
$ G_\csf\wedge (W_\asf\rightarrow G_\bsf) $ & $ \vdash_{\mathcal{L}}$ & $ (W_\asf\wedge W_\bsf)\wedge (W_\asf\rightarrow G_\bsf)$\\
& $ \vdash_{\mathcal{L}}$ & $ W_\bsf\wedge G_\bsf$\\
& $ \vdash_{\mathcal{L}}$ & $ \bot,$\\
\end{tabular}
\end{center}
and likewise $ G_\bsf\wedge (W_\asf\rightarrow G_\csf)\vdash_{\mathcal{L}} \bot$  (which together justify the entailment marked with ($\sim$) below); by FS1 and Fact \ref{fct: FS1 equivalently}, the entailments marked with ($\ast$) hold in the following chain, and $\mathsf{aut}$ and $\mathsf{one}$ imply that $ W_\asf\dashv\vdash_{\mathcal{L}} (G_\bsf\vee G_\csf)$
(which justifies the entailment marked with ($\ast\ast$) below); hence:
\begin{center}
\begin{tabular}{r l l l}
&& $\Box_\csf(W_\asf\rightarrow (\Diamond_\asf G_\bsf\wedge \Diamond_\asf G_\csf))\wedge \Diamond_\csf [W_\asf\wedge [(G_\bsf\rightarrow \Box_\asf (W_\asf\rightarrow G_\bsf))\wedge (G_\csf\rightarrow \Box_\asf (W_\asf\rightarrow G_\csf))]]$ &\\
($\ast$) & $\vdash_{IEAK_I}$& $\Diamond_\csf[W_\asf\wedge (W_\asf\rightarrow (\Diamond_\asf G_\bsf\wedge \Diamond_\asf G_\csf))\wedge [(G_\bsf\rightarrow \Box_\asf (W_\asf\rightarrow G_\bsf))\wedge (G_\csf\rightarrow \Box_\asf (W_\asf\rightarrow G_\csf))]]$&\\
($\ast\ast$)& $\vdash_{\mathcal{L}}$& $\Diamond_\csf[(G_\bsf\vee G_\csf)\wedge  (\Diamond_\asf G_\bsf\wedge \Diamond_\asf G_\csf)\wedge [(G_\bsf\rightarrow \Box_\asf (W_\asf\rightarrow G_\bsf))\wedge (G_\csf\rightarrow \Box_\asf (W_\asf\rightarrow G_\csf))]]$&\\
& $\vdash_{IEAK_I}$& $\Diamond_\csf[G_\bsf\wedge  \Diamond_\asf G_\csf\wedge (G_\bsf\rightarrow \Box_\asf (W_\asf\rightarrow G_\bsf))]\vee \Diamond_\csf[G_\csf\wedge  \Diamond_\asf G_\bsf\wedge (G_\csf\rightarrow \Box_\asf (W_\asf\rightarrow G_\csf))]$&\\
& $\vdash_{IEAK_I}$& $\Diamond_\csf[\Diamond_\asf G_\csf\wedge  \Box_\asf (W_\asf\rightarrow G_\bsf)]\vee \Diamond_\csf[ \Diamond_\asf G_\bsf\wedge  \Box_\asf (W_\asf\rightarrow G_\csf)]$&\\
($\ast$)& $\vdash_{IEAK_I}$& $\Diamond_\csf\Diamond_\asf (G_\csf\wedge (W_\asf\rightarrow G_\bsf))\vee \Diamond_\csf\Diamond_\asf (G_\bsf\wedge (W_\asf\rightarrow G_\csf))$&\\
($\sim$)& $\vdash_{IEAK_I}$& $\Diamond_\csf\Diamond_\asf \bot\vee \Diamond_\csf\Diamond_\asf \bot$&\\
& $\vdash_{IEAK_I}$& $\bot$.&\\
\end{tabular}
\end{center}
\end{proof}

\begin{remark} It may be helpful to compare the proof above both with
  the informal argument and with a semantic proof.
\begin{enumerate}
\item The informal proof goes as follows. After the action $\alpha$,
  agent $\csf$ knows that either\\
  -- $\asf$ knows who has the green card, this being the case iff $a$ holds the green card herself, or\\
  -- $\asf$ doesn't know who has the green card, this being the case iff
  $a$ doesn't hold the green card.\\
  After the public announcement $\beta$ of $\asf$ knowing who has the
  green card, agent $\csf$ can discard the second alternative and
  conclude from the first one that $\asf$ holds the green card.
\item Comparing the formal and the informal proof, we see that the formal proof roughly follows the same structure. In the formal proof, although tedious, all the steps discharging
  (1) and (2) are routine. Proving (3), however, corresponds to agent
  $c$ reasoning that after announcement of $\beta$ the second
  alternative of the item above cannot hold. And indeed, our formal
  proof proceeds by deriving a contradiction from the assumption that,
  after $\alpha$, agent $\csf$ thinks it is possible to be in a state
  where $\asf$ does not know who has the green card.
\item The use of contradiction in our formal proof does not violate the laws of
  intuitionistic logic (ex-falso-quodlibet is intuitionistically
  valid). But we use that, according to $\mathsf{aut}$ and
  $\mathsf{one}$, the atomic propositions $W_\isf, G_\isf$ behave as the Boolean negations of one another, for each agent $i$.

\item
  A semantic proof would typically start from a Kripke model $M$ capturing the situation  described at the beginning of the section. For example, $M$ could have three states corresponding to the three possibilities of who holds the green card (see Example~\ref{exle:actions} for pictures); moreover, the two states in which $G_\bsf$ and respectively $G_\csf$ holds would be indistinguishable for $\asf$, with similar indistinguishability relations holding for agents $\bsf$ and $\csf$. Next, we can compute $M^\alpha$ which is as $M$ but with  a $\bsf$-edge deleted, as now $\bsf$ knows who has the green card. Finally, we compute $(M^\alpha)^\beta$ and check that it consists of a single state in which $G_\asf$ holds, proving that now everybody knows that $a$ holds the green card.
\item Comparing our formal proof with the semantic argument, the proof theoretic argument has the advantage that it establishes the result not only for one model, but for all models satisfying $\mathsf{aut}$, $\mathsf{one}$, and $E(\mathsf{other?})$. It is thus revealed, for example, that the argument does not require that knowledge is encoded by an equivalence relation or that is satisfies introspection $\Box p\to p$.
\end{enumerate}
\end{remark}

\section{Conclusion}

The application of duality theory to dynamic epistemic logic begun in \cite{MPS} for the logic of public announcements and, generalized here to Baltag-Moss-Solecki's logic of Epistemic Actions and Knowledge, opens new directions of research which we plan to pursue in the future.

First, as mentioned in the introduction, the generalization of modal logic  to coalgebraic logic can be cast in the framework of duality theory; hence, the results of the present paper naturally link up  with a line of research in the coalgebraic theory of epistemic updates which has its precursor in \cite{Ger99} and further explored in \cite{Ba03, CS07}. We plan to further explore this link, both to export the technique of dynamic updates from Kripke frames to coalgebras, and to make  coalgebraic techniques bear on variations of the Kripke semantics of \cite{BMS, BMSj} to a variety of semantic scenarios based on, for example, probabilistic or neighborhood semantics. Moreover,  the fruitfulness of the coalgebraic point of view on epistemic actions is also emphasized by the fact that certain aspects of dynamic (epistemic) logics are most easily understood by considering their semantics not in general models but in the final coalgebra, as discussed in \cite{Ba03,CS07}.

Second, we plan to explore the generalization of dynamic epistemic logics from classical to nonclassical logic. On the one hand, general observations indicate that  `dynamic phenomena' are in many important contexts best analyzed using an appropriate nonclassical  logic; for instance, in all those contexts (such as scientific experiments, acquisition of legal evidence, verification of programs, etc.) where the notion of truth is {\em procedural}. In these contexts, affirming $\phi$  means demonstrating that {\em some} appropriate {\em instance} of the procedure applies to $\phi$; 
refuting $\phi$ means demonstrating that {\em some} appropriate {\em instance} of the procedure applies to $\neg \phi$; however, neither instance might be available in some cases, hence the law of excluded middle fails. In these situations, intuitionistic or weaker logics provide  viable alternatives.

On the other hand, computer science offers a considerable number of intuitionistic modal logics which might be extended to  dynamic versions. For example, the lax logic of Fairtlough and Mendler \cite{FM97} has been proposed for hardware verification, but since then resurfaced in quite different scenarios. Furthermore, logics for access control tend to be intuitionistic \cite{A08,GN08} as well as logics used for agreeing contracts in web services as in propositional contract logic \cite{BZ10}. Other interesting instances deserving study are dynamic updates on a linear propositional base, (e.g.\ taking quantales as underlying algebras) or on a quantum base (taking orthomodular lattices as underlying algebras).

Closely connected to the previous point is the third direction to be pursued, concerning {\em proof systems} for dynamic logics. In collaboration with Giuseppe Greco, we are developing sound, complete and cut-free display-style sequent calculi for the intuitionistic and the classical versions of PAL and EAK (see \cite{GKP13a, GKP13b}). The choice of the display calculi format allows for a great degree of modularity. We expect that these calculi will lend themselves very well to provide a uniform account  of the further developments outlined in the previous direction.

\section{Appendix}

\subsection{HA- and FSA-identities and inequalities}
In a Heyting algebra $\wedge$ and $\to$ are {\em residuated}, namely, for all $x, y, z\in \bbA$,
\begin{equation}\label{eq:HA}x\wedge y\leq z\quad \mbox{ iff }\quad x\leq y\to z.\end{equation}
Hence, by the general theory of residuation,
\begin{equation}\label{eq:H arrow} y\to z = \bigvee\{x\mid x\wedge y\leq z\}.\end{equation}

Using (\ref{eq:HA}) and (\ref{eq:H arrow}) above, it is not difficult to prove the following
\begin{fact}
\label{fct:HA}
For every Heyting algebra $\bbA$ and all $x, y, z\in \bbA$,
\begin{enumerate}
\item $x\wedge (x\to y)\leq y$.
\item $x\to (y\wedge z) = (x\to y)\wedge (x\to z)$.
\item $x\wedge y\leq x\to y$.
\item $x\to y  = x\to (x\wedge y)$.
\item $(x\wedge y)\to z = x\to (y\to z)$.
\item $x\wedge (y\to z) = x\wedge ((x\wedge y)\to z)$.
\end{enumerate}
\end{fact}

\begin{fact}
\label{fct: FS1 equivalently}
The following are provably equivalent in IK:
\begin{enumerate}
\item $\Diamond(p\to q)\leq \Box p\to \Diamond q$;

\item $\Box p\wedge \Diamond q\leq \Diamond(p\wedge q)$;

\item $\Box (p\to q)\leq \Diamond p\to \Diamond q$.
\end{enumerate}
\end{fact}

\commment{
\begin{proposition}
\label{prop: FSA and FSA black}
For every perfect HA $\bbA$, its modal expansion $(\bbA, \Diamond, \Box)$ is a perfect FSA iff  $(\bbA, \Diamondblack, \blacksquare)$ is.
\end{proposition}
}

\subsection{Properties of the map $i'$}
The following fact is a straightforward specialization of \cite[Fact 28]{MPS}.
\begin{fact}
\label{fct:i'}
Let $\bbA$ be an FS-/MIPC-algebra, $a$ be an action structure over $\bbA$, and let
$i': \bbA^a\to \prod_a\bbA$ given by $[b]\mapsto b\wedge Pre_a$. Then, for every $b, c\in \bbA^a$,
\begin{enumerate}
\item $i'(b\vee c) = i'(b)\vee i'(c)$;

\item   $i'(b\wedge c) = i'(b)\wedge i'(c)$;

\item  $i'(b\to c) = Pre_a\wedge (i'(b)\to i'(c))$;

\item $ i'( \Diamond^a b)= \Diamond^{\prod_a\bbA}(i'(b)\wedge Pre_a) \wedge Pre_a$;

\item $ i'( \Box^a b)= Pre_a\to \Box^{\prod_a\bbA}(Pre_a\to i'(b))$.
\end{enumerate}
\end{fact}

\subsection{Soundness Lemmas}
In the present subsection, the lemmas are collected which serve to prove Proposition \ref{prop:IPAL soundness}.
\begin{lemma}
\label{fct: preservation of facts}
Let $M = (\bbA, V)$ be an algebraic model and let $\alpha$ be an action structure over $\mathcal{L}$. For every formula $\phi$ such that\ $\val{\phi}_{M^\alpha} = \pi(\val{\phi}_{\prod_\alpha M})$,
\begin{enumerate}
\item  $\val{ \langle \alpha \rangle \phi}_M = \val{Pre(\alpha)}_M \wedge  \val{\phi}_{M}.$
\item $\val{ [\alpha] \phi}_M = \val{Pre(\alpha)}_M \to  \val{\phi}_{M}.$
\end{enumerate}
\end{lemma}
\begin{proof}\hfill
\begin{enumerate}
\item
\[
\begin{tabu}{r c l l}
&&\val{ \langle \alpha \rangle \phi}_M \\
&=&\val{Pre(\alpha)}_M \wedge  \pi_k\circ i'(\val{\phi}_{M^\alpha})\\
&=&\val{Pre(\alpha)}_M \wedge  \pi_k\circ i'(\pi(\val{\phi}_{\prod_\alpha M}))\\
&=&\val{Pre(\alpha)}_M \wedge  \pi_k(\val{\phi}_{\prod_\alpha M}\wedge Pre_\alpha)\\
&=&\val{Pre(\alpha)}_M \wedge  (\pi_k(\val{\phi}_{\prod_\alpha M})\wedge \pi_k(Pre_\alpha))\\
&=&\val{Pre(\alpha)}_M \wedge  (\val{\phi}_{M}\wedge Pre_\alpha(k))\\
&=&\val{Pre(\alpha)}_M \wedge  (\val{\phi}_{M}\wedge \val{Pre(\alpha)}_M)\\
&=&\val{Pre(\alpha)}_M \wedge  \val{\phi}_M.&\mbox{\phantom{(Fact \ref{fct:HA}.4)}}
\end{tabu}
\]
\item
\[\begin{tabu}{r c l l}
&&\val{[\alpha]\phi}_M\\ 
&=&\val{Pre(\alpha)}_M \to  \pi_k\circ i'(\val{\phi}_{M^\alpha})&\\
&=&\val{Pre(\alpha)}_M \to  \pi_k\circ i'(\pi(\val{\phi}_{\prod_\alpha M}))&\\
&=&\val{Pre(\alpha)}_M \to  \pi_k(\val{\phi}_{\prod_\alpha M}\wedge Pre_\alpha)&\\
&=&\val{Pre(\alpha)}_M \to  (\pi_k(\val{\phi}_{\prod_\alpha M})\wedge \pi_k(Pre_\alpha))&\\
&=&\val{Pre(\alpha)}_M \to (\val{\phi}_{M}\wedge Pre_\alpha(k))&\\
&=&\val{Pre(\alpha)}_M \to  (\val{\phi}_{M}\wedge \val{Pre(\alpha)}_M)&\\
&=&\val{Pre(\alpha)}_M \to  \val{\phi}_M &\mbox{(Fact \ref{fct:HA}.4)}.
\end{tabu}
\]
\end{enumerate}
\vspace{-\baselineskip}\end{proof}

\begin{lemma}
\label{fct:normality of dynamic modalities}
Let $M = (\bbA, V)$ be an algebraic model.
For every  action structure $\alpha$  over $\mathcal{L}$ and all formulas $\phi$ and $\psi$,
\begin{enumerate}
\item $\val{ \langle  \alpha \rangle  (\phi\vee \psi)}_M  = \val{\langle \alpha\rangle \phi}_M \vee  \val{\langle\alpha \rangle\psi}_{M}$. 
\item $\val{ [\alpha]  (\phi\vee \psi)}_M  = \val{Pre(\alpha)}_M \rightarrow  (\val{\langle \alpha \rangle \phi}_M \vee \val{\langle \alpha \rangle \psi}_M)$.
\end{enumerate}
\end{lemma}

\begin{proof}\hfill
\begin{enumerate}
\item
\[
\begin{tabu}{r c l l}
&&\val{ \langle  \alpha \rangle  (\phi\vee \psi)}_M \\
&=&\val{Pre(\alpha)}_M \wedge \pi_k\circ i'(\val{\phi\vee \psi }_{M^\alpha})&\\
&=&\val{Pre(\alpha)}_M \wedge (\pi_k\circ i'(\val{\phi}_{M^\alpha})\vee \pi_k\circ i'(\val{\psi }_{M^\alpha}))&\mbox{(Fact \ref{fct:i'}.1)}\\
&=&(\val{Pre(\alpha)}_M \wedge \pi_k\circ i'(\val{\phi}_{M^\alpha}))\vee (\val{Pre(\alpha)}_M \wedge \pi_k\circ i'(\val{\psi }_{M^\alpha})))& \\
&=&\val{\langle \alpha\rangle \phi}_M \vee  \val{\langle\alpha \rangle\psi}_{M}. &
\end{tabu}
\]

\item
\[\begin{tabu}{r c l l}
&&\val{[\alpha](\phi \vee \psi)}_M\\
&=&\val{Pre(\alpha)}_M \rightarrow \pi_k\circ i'(\val{\phi \vee \psi}_{M^\alpha}) \\
&=&\val{Pre(\alpha)}_M \rightarrow (\pi_k\circ i'(\val{\phi}_{M^\alpha}) \vee \pi_k\circ i'(\val{\psi}_{M^\alpha}))&\mbox{(Fact \ref{fct:i'}.1)}\\
&=&\val{Pre(\alpha)}_M \rightarrow (\val{Pre(\alpha)}_M \wedge (\pi_k\circ i'(\val{\phi}_{M^\alpha}) \vee \pi_k\circ i'(\val{\psi}_{M^\alpha})))&\mbox{(Fact \ref{fct:HA}.4)}\\
&=&\val{Pre(\alpha)}_M \rightarrow  ((\val{Pre(\alpha)}_M  \wedge \pi_k\circ i'(\val{\phi}_{M^\alpha}))\vee (\val{Pre(\alpha)}_M  \wedge \pi_k\circ i'(\val{\psi}_{M^\alpha})))&\\
&=&\val{Pre(\alpha)}_M \rightarrow  (\val{\langle \alpha \rangle \phi}_M \vee \val{\langle \alpha \rangle \psi}_M).&\\
\end{tabu}
\]
\end{enumerate}
\vspace{-\baselineskip}\end{proof}

\begin{lemma}
\label{fct:uncongenial distribution over lattice struct}
Let $M = (\bbA, V)$ be an algebraic model.
For every  action structure $\alpha$  over $\mathcal{L}$ and all formulas $\phi$ and $\psi$,
\begin{enumerate}
\item $\val{ \langle  \alpha \rangle  (\phi\wedge \psi)}_M  = \val{\langle \alpha\rangle \phi}_M \wedge  \val{\langle\alpha \rangle\psi}_{M}$.
\item $\val{ [\alpha]  (\phi\wedge \psi)}_M  = \val{[\alpha] \phi}_M \wedge \val{[\alpha]\psi}_{M}$.
\end{enumerate}
\end{lemma}
\begin{proof}\hfill
\begin{enumerate}
\item
\[\begin{tabu}{r c l l}
&&\val{\langle \alpha \rangle (\phi \wedge \psi)}_M \\
&=&\val{Pre(\alpha)}_M \wedge \pi_k\circ i'(\val{\phi \wedge \psi}_{M^\alpha})&\\
&=&\val{Pre(\alpha)}_M \wedge (\pi_k\circ i'(\val{\phi}_{M^\alpha}) \wedge \pi_k\circ i'(\val{\psi}_{M^\alpha}))&\mbox{(Fact \ref{fct:i'}.2)}\\
&=&(\val{Pre(\alpha)}_{M^\alpha} \wedge \pi_k\circ i'(\val{\phi}_{M^\alpha})) \wedge (\val{Pre(\alpha)}_M \wedge \pi_k\circ i'(\val{\psi}_{M^\alpha}))&\\
&=&\val{\langle \alpha \rangle \phi}_M \wedge \val{\langle \alpha \rangle \psi}_M.&
\end{tabu}
\]

\item
\[\begin{tabu}{r c l l}
&&\val{ [\alpha ]  (\phi\wedge \psi)}_M \\
&=&\val{Pre(\alpha)}_M \to \pi_k\circ i'(\val{\phi\wedge \psi }_{M^\alpha})&\\
&=&\val{Pre(\alpha)}_M \to \pi_k\circ i'(\val{\phi}_{M^\alpha}\wedge \val{\psi}_{M^\alpha})&\\
&=&\val{Pre(\alpha)}_M \to (\pi_k\circ i'(\val{\phi}_{M^\alpha})\wedge \pi_k\circ i'(\val{\psi }_{M^\alpha}))&\mbox{(Fact \ref{fct:i'}.2)}\\
&=&(\val{Pre(\alpha)}_M \to \pi_k\circ i'(\val{\phi}_{M^\alpha}))\wedge (\val{\alpha}_M \to \pi_k\circ i'(\val{\psi }_{M^\alpha}))&\mbox{(Fact \ref{fct:HA}.2)}\\
&=&\val{[\alpha] \phi}_M \wedge  \val{[\alpha]\psi}_{M}.
\end{tabu}
\]
\end{enumerate}
\vspace{-\baselineskip}\end{proof}

\begin{lemma}
\label{fct:distribution over implication}
Let $M = (\bbA, V)$ be an algebraic model.
For every  action structure $\alpha$  over $\mathcal{L}$ and all formulas $\phi$ and $\psi$,
\begin{enumerate}
\item $\val{ [\alpha ] (\phi \to \psi)}_M = \val{\langle \alpha \rangle \phi}_M \to\val{\langle \alpha \rangle \psi}_M$.
\item $\val{\langle \alpha \rangle (\phi \to \psi)}_M  = \val{Pre(\alpha)}_M \wedge (\val{\langle \alpha \rangle \phi}_M \to\val{\langle \alpha \rangle \psi}_M)$.
\end{enumerate}
\end{lemma}
\begin{proof}
 We preliminarily observe that
\begin{center}
\begin{tabular}{r c l l}
&  & $ (\val{Pre(\alpha)}_M \wedge  \pi_k\circ i'(\val{\phi}_{M^\alpha})) \to \pi_k\circ i'(\val{\psi}_{M^\alpha})$& \\
& $=$ &$ (\val{Pre(\alpha)}_M \wedge  \pi_k\circ i'(\val{\phi}_{M^\alpha})) \to ((\val{Pre(\alpha)}_M \wedge  \pi_k\circ i'(\val{\phi}_{M^\alpha}))\wedge \pi_k\circ i'(\val{\psi}_{M^\alpha}))$& (Fact \ref{fct:HA}.4)\\
& $=$ &$ \val{\langle \alpha \rangle \phi}_M \to( \val{\langle \alpha \rangle \phi}_M \wedge \val{\langle \alpha \rangle \psi}_M)$&\\
& $=$ &$ \val{\langle \alpha \rangle \phi}_M \to\val{\langle \alpha \rangle \psi}_M$.&(Fact \ref{fct:HA}.4)\\
\end{tabular}
\end{center}
Hence:
\begin{enumerate}
\item
\[\begin{tabu}{r c l l}
&&\val{ [\alpha ] (\phi \to \psi)}_M \\
&=&\val{Pre(\alpha)}_M \to \pi_k\circ i'(\val{\phi \to \psi}_{M^\alpha})&\\
&=&\val{Pre(\alpha)}_M \to \pi_k (Pre_\alpha \wedge (i'(\val{\phi}_{M^\alpha}) \to i'(\val{\psi}_{M^\alpha})))&\mbox{(Fact \ref{fct:i'}.3)}\\
&=&\val{Pre(\alpha)}_M \to (Pre_\alpha(k) \wedge (\pi_k\circ i'(\val{\phi}_{M^\alpha}) \to\pi_k\circ  i'(\val{\psi}_{M^\alpha})))& \\
&=&\val{Pre(\alpha)}_M \to (\val{Pre(\alpha)}_M \wedge (\pi_k\circ i'(\val{\phi}_{M^\alpha}) \to\pi_k\circ  i'(\val{\psi}_{M^\alpha})))& \\
&=&\val{Pre(\alpha)}_M \to (\pi_k\circ i'(\val{\phi}_{M^\alpha}) \to \pi_k\circ i'(\val{\psi}_{M^\alpha}))&\mbox{(Fact \ref{fct:HA}.4)}\\
&=&(\val{Pre(\alpha)}_M \wedge \pi_k\circ  i'(\val{\phi}_{M^\alpha})) \to \pi_k\circ  i'(\val{\psi}_{M^\alpha})&\mbox{(Fact \ref{fct:HA}.5)}\\
&=&\val{\langle \alpha \rangle \phi}_M \to\val{\langle \alpha \rangle \psi}_M.&
\end{tabu}
\]
\item
\[\begin{tabu}{r c l l}
&&\val{\langle \alpha \rangle (\phi \to \psi)}_M \\
&=&\val{Pre(\alpha)}_M \wedge \pi_k\circ i'(\val{\phi \to \psi}_{M^\alpha})&\\
&=&\val{Pre(\alpha)}_M \wedge \pi_k(Pre_\alpha \wedge (i'(\val{\phi}_{M^\alpha}) \to i'(\val{\psi}_{M^\alpha})))&\mbox{(Fact \ref{fct:i'}.3)}\\
&=&\val{Pre(\alpha)}_M \wedge (Pre_\alpha(k) \wedge (\pi_k\circ i'(\val{\phi}_{M^\alpha}) \to \pi_k\circ i'(\val{\psi}_{M^\alpha})))& \\
&=&\val{Pre(\alpha)}_M \wedge (\val{Pre(\alpha)}_M \wedge (\pi_k\circ i'(\val{\phi}_{M^\alpha}) \to \pi_k\circ i'(\val{\psi}_{M^\alpha})))& \\
&=&\val{Pre(\alpha)}_M \wedge (\pi_k\circ i'(\val{\phi}_{M^\alpha}) \to \pi_k\circ i'(\val{\psi}_{M^\alpha}))&\mbox{(Fact \ref{fct:HA}.4)}\\ 
&=&\val{Pre(\alpha)}_M \wedge ((\val{Pre(\alpha)}_M \wedge \pi_k\circ i'(\val{\phi}_{M^\alpha})) \to \pi_k\circ i'(\val{\psi}_{M^\alpha}))&\mbox{(Fact \ref{fct:HA}.6)}\\
&=&\val{Pre(\alpha)}_M \wedge (\val{\langle \alpha \rangle \phi}_M \to\val{\langle \alpha \rangle \psi}_M).\\
\end{tabu}
\]
\end{enumerate}
\vspace{-\baselineskip}\end{proof}

\begin{fact}
\label{fct:Malpha =Malphaj}
Let $M = (\bbA, V)$ be an algebraic model, and let $\alpha = (K, k, \alpha, Pre_\alpha)$ be an  action structure  over $\mathcal{L}$. For every $j\in K$, $$M^\alpha = M^{\alpha_j}.$$
\end{fact}
\begin{proof}
Recall that $\alpha_j: = (K, j, \alpha, Pre_\alpha)$. The statement immediately follows from the observation that no component of the definition of the updated model $M^\alpha$ (cf.\ Definition \ref{def:update model}) depends on the designated element in the action structure $\alpha$.
\end{proof}

\begin{lemma}
\label{fct:distribution over diamond}
Let $M = (\bbA, V)$ be an algebraic model.
For every  action structure $\alpha$  over $\mathcal{L}$ and every formula $\phi$,
\begin{enumerate}
\item $\val{\langle\alpha\rangle\Diamond \phi}_M = \val{Pre(\alpha)}_M\wedge \bigvee\{\Diamond^\bbA (\val{\langle\alpha_j\rangle \phi}_{M})\mid k\alpha j\}.$
\item $\val{[\alpha]\Diamond \phi}_M = \val{Pre(\alpha)}_M\to \bigvee\{\Diamond^\bbA (\val{\langle\alpha_j\rangle \phi}_{M})\mid k\alpha j\}.$
\end{enumerate}
\end{lemma}\newpage

\begin{proof} We  preliminarily observe that
\begin{center} \begin{tabular}{r c l l}
&&$\pi_k\circ i'(\val{\Diamond \phi}_{M^\alpha})$\\ &$=$&$\pi_k(Pre_\alpha\wedge \Diamond^{\prod_\alpha\bbA} (Pre_\alpha\wedge i'(\val{\phi}_{M^\alpha})))$& (Fact \ref{fct:i'}.4)\\
&$=$&$ Pre_\alpha(k)\wedge \bigvee\{\Diamond^\bbA (Pre_\alpha\wedge i'(\val{\phi}_{M^\alpha}))(j)\mid k\alpha j\}$& \eqref{eq:diam-prod} \\
&$=$&$ \val{Pre(\alpha)}_M\wedge \bigvee\{\Diamond^\bbA (Pre_\alpha(j)\wedge i'(\val{\phi}_{M^\alpha})(j))\mid k\alpha j\}$&\\
&$=$&$ \val{Pre(\alpha)}_M\wedge \bigvee\{\Diamond^\bbA (\val{Pre(\alpha_j)}_M\wedge \pi_j\circ i'(\val{\phi}_{M^\alpha}))\mid k\alpha j\}$&\\
&$=$&$ \val{Pre(\alpha)}_M\wedge \bigvee\{\Diamond^\bbA (\val{Pre(\alpha_j)}_M\wedge \pi_j\circ i'(\val{\phi}_{M^{\alpha_j}}))\mid k\alpha j\}$& (Fact \ref{fct:Malpha =Malphaj})\\
&$=$&$ \val{Pre(\alpha)}_M\wedge \bigvee\{\Diamond^\bbA (\val{\langle\alpha_j\rangle \phi}_{M})\mid k\alpha j\}$.&\\
\end{tabular}\end{center}
Hence:
\begin{enumerate}
\item
\[\begin{tabu}{r c l l}
&&\val{ \langle  \alpha \rangle \Diamond\phi}_M \\
&=&\val{Pre(\alpha)}_M \wedge \pi_k\circ i'(\val{\Diamond\phi }_{M^\alpha})&\\
&=&\val{Pre(\alpha)}_M \wedge (\val{Pre(\alpha)}_M\wedge \bigvee\{\Diamond^\bbA (\val{\langle\alpha_j\rangle \phi}_{M})\mid k\alpha j\})& \\
&=&\val{Pre(\alpha)}_M\wedge \bigvee\{\Diamond^\bbA (\val{\langle\alpha_j\rangle \phi}_{M})\mid k\alpha j\}.
\end{tabu}
\]
\item
\[\begin{tabu}{r c l l}
&&\val{ [\alpha] \Diamond\phi}_M \\
&=&\val{Pre(\alpha)}_M \to i'(\val{\Diamond\phi }_{M^\alpha})&\\
&=&\val{Pre(\alpha)}_M \to (\val{Pre(\alpha)}_M\wedge \bigvee\{\Diamond^\bbA (\val{\langle\alpha_j\rangle \phi}_{M})\mid k\alpha j\})& \\
&=&\val{Pre(\alpha)}_M \to \bigvee\{\Diamond^\bbA (\val{\langle\alpha_j\rangle \phi}_{M})\mid k\alpha j\}.&\mbox{(Fact \ref{fct:HA}.4)}
\end{tabu}
\]
\end{enumerate}
\vspace{-\baselineskip}\end{proof}

\begin{lemma}
\label{fct:distribution over box}
Let $M = (\bbA, V)$ be an algebraic model.
For every  action structure $\alpha$  over $\mathcal{L}$ and every formula  $\phi$,
\begin{enumerate}
\item $\val{\langle\alpha\rangle\Box \phi}_M = \val{Pre(\alpha)}_M\wedge \bigwedge\{\Box^\bbA (\val{[\alpha_j] \phi}_{M})\mid k\alpha j\}.$
\item $\val{[\alpha]\Box \phi}_M = \val{Pre(\alpha)}_M\to \bigwedge\{\Box^\bbA (\val{[\alpha_j] \phi}_{M})\mid k\alpha j\}.$
\end{enumerate}
\end{lemma}
\begin{proof}
We preliminarily observe that
\begin{center} \begin{tabular}{r c l l}
&&$\pi_k\circ i'(\val{\Box \phi}_{M^\alpha})$\\ &$=$&$\pi_k(Pre_\alpha\to \Box^{\prod_\alpha\bbA} (Pre_\alpha\to i'(\val{\phi}_{M^\alpha})))$& (Fact \ref{fct:i'}.5)\\
&$=$&$ Pre_\alpha(k)\to \bigwedge\{\Box^\bbA (Pre_\alpha\to i'(\val{\phi}_{M^\alpha}))(j)\mid k\alpha j\}$&\eqref{eq:box-prod}\\
&$=$&$ \val{Pre(\alpha)}_M\to \bigwedge\{\Box^\bbA (Pre_\alpha(j)\to i'(\val{\phi}_{M^\alpha})(j))\mid k\alpha j\}$&\\
&$=$&$ \val{Pre(\alpha)}_M\to \bigwedge\{\Box^\bbA (\val{Pre(\alpha_j)}_M\to \pi_j\circ i'(\val{\phi}_{M^\alpha}))\mid k\alpha j\}$&\\
&$=$&$ \val{Pre(\alpha)}_M\to \bigwedge\{\Box^\bbA (\val{Pre(\alpha_j)}_M\to \pi_j\circ i'(\val{\phi}_{M^{\alpha_j}}))\mid k\alpha j\}$&(Fact \ref{fct:Malpha =Malphaj})\\
&$=$&$ \val{Pre(\alpha)}_M\to \bigwedge\{\Box^\bbA (\val{[\alpha_j] \phi}_{M})\mid k\alpha j\}$.&\\
\end{tabular}\end{center}

Hence:
\begin{enumerate}
\item
\[\begin{tabu}{r c l l}
&&\val{ \langle  \alpha \rangle \Box\phi}_M \\
&=&\val{Pre(\alpha)}_M \wedge \pi_k\circ i'(\val{\Box\phi }_{M^\alpha})&\\
&=&\val{Pre(\alpha)}_M \wedge (\val{Pre(\alpha)}_M\to \bigwedge\{\Box^\bbA (\val{[\alpha_j] \phi}_{M})\mid j\alpha k\})& \\
&=&\val{Pre(\alpha)}_M \wedge \bigwedge\{\Box^\bbA (\val{[\alpha_j] \phi}_{M})\mid j\alpha k\}.
\end{tabu}
\]
\item
\[\begin{tabu}{r c l l}
&&\val{[\alpha]\Box \phi}_M \\
&=&\val{\alpha}_M \to \pi_k\circ i'(\val{\Box\phi}_{M^\alpha})&\\
&=&\val{\alpha}_M \to (\val{Pre(\alpha)}_M\to \bigwedge\{\Box^\bbA (\val{[\alpha_j] \phi}_{M})\mid j\alpha k\}) &\\
&=&\val{\alpha}_M \to \bigwedge\{\Box^\bbA (\val{[\alpha_j] \phi}_{M})\mid j\alpha k\}.  &\mbox{(Fact \ref{fct:HA}.4)}
\end{tabu}
\]
\end{enumerate}
\vspace{-\baselineskip}\end{proof}

\end{document}